\documentclass{article}

 \usepackage[citestyle=numeric,bibstyle=authoryear]{biblatex}

\DeclareNameAlias{sortname}{first-last}
\DeclareFieldFormat{labelnumberwidth}{\mkbibbrackets{#1}}
\defbibenvironment{bibliography}
  {\list
     {\printtext[labelnumberwidth]{%
    \printfield{prefixnumber}%
    \printfield{labelnumber}}}
     {\setlength{\labelwidth}{\labelnumberwidth}%
      \setlength{\leftmargin}{\labelwidth}%
      \setlength{\labelsep}{\biblabelsep}%
      \addtolength{\leftmargin}{\labelsep}%
      \setlength{\itemsep}{\bibitemsep}%
      \setlength{\parsep}{\bibparsep}}%
      }
  {\endlist}
  {\item}

\usepackage[utf8]{inputenc}
\usepackage{hyperref}
\usepackage{pgfgantt}
\hypersetup{
    colorlinks=true,
    linkcolor=blue,
    filecolor=magenta,      
    urlcolor=blue,
    citecolor=blue,
    pdftitle={Overleaf Example},
    pdfpagemode=FullScreen,
    }
\usepackage[a4paper, total={6in, 8in}]{geometry}

\usepackage{amsmath,amssymb,amsfonts}
\usepackage{graphicx}
\graphicspath{ {./images/} }
\newtheorem{theorem}{Theorem}[section]
\newtheorem{proposition}[theorem]{Proposition}
\newtheorem{corollary}[theorem]{Corollary}
\newtheorem{lemma}[theorem]{Lemma}

\newtheorem{definition}[theorem]{Definition}
\newtheorem{remark}[theorem]{Remark}
\newtheorem{notation}[theorem]{Notation}
\newtheorem{example}[theorem]{Example}

\newenvironment{proof}{\mbox{\bf Proof.}}{\mbox{$\Box$}\bigskip}

 \usepackage{authblk}
\begin{document}

\title{The Division Problem of Chances\footnote{We gratefully acknowledge support from the Swiss National Science
Foundation (SNFS) through Project 100018$-$212311.}}
 \author{ Rasoul Ramezanian 
\thanks{  I would like to express my  gratitude to Bettina Klaus for her excellent guidance, feedback, and our insightful discussions, which have significantly enhanced the quality of this work. Additionally, I extend my appreciation to Arunava Sen for his invaluable and constructive feedback. }
}
\date{}
\affil{Faculty of Business and Economics

University of Lausanne

Internef, 1015 Lausanne,
Switzerland

 \href{mailto: rasoul.ramezanian@unil.ch}{rasoul.ramezanian@unil.ch}} 
\maketitle

\vspace{-0.9cm}
 \begin{abstract} 
	\noindent  In frequently repeated matching scenarios, individuals may require diversification in their choices. Therefore, when faced with a set of potential outcomes, each individual may have  an ideal \hspace{0.5cm} lottery over outcomes  that represents their preferred option. This suggests that, as people seek variety, their favorite choice is not a particular outcome, but rather a lottery over them as their peak for their preferences.  
  We explore matching problems in situations where agents' preferences are represented by ideal lotteries. Our focus lies in addressing the challenge of dividing chances in matching, where agents express their preferences over a set of objects through ideal lotteries that reflect their single-peaked preferences.
   We discuss properties such as  strategy proofness, replacement monotonicity, (Pareto) efficiency, in-betweenness, non-bossiness, envy-freeness, and anonymity in the context of dividing chances, and   propose a class of mechanisms called URC mechanisms that satisfy these properties. Subsequently, we prove that  if a mechanism for dividing chances is strategy proof, (Pareto) efficient, replacement monotonic, in-between,  non-bossy,  and anonymous (or envy free), then it is equivalent in terms of welfare to a URC mechanism.
\end{abstract}



\noindent \textbf{MSC2020 Subject Classification}: 9110, 91B32, 91B03, 91B08, 91B68

\noindent  \textbf{Keywords}: Matching Theory, Fair Division, Single peaked Preferences,  Mechanism Design. 

\newpage

\section{Introduction}\label{sec:intro}

  Suppose you have a collection of music records on your smartphone, and there is an application on your device that allows you to set how frequently you want certain songs to be played \textbf{repeatedly}. For instance, you could set your ideal lottery as $c=(a:0.2,b:0.5,c:0.1,d:0.2)$, which indicates that you want to listen to  record $a$, 20 percent of the time,  record $b$, 50 percent of the time,  record $c$, 10 percent of the time, and  record, $d$, 20 percent of the time repeatedly. The application repeatedly executes your ideal lottery and selects a record to play for you. Your favorite option is not a specific music record (if it were, then your collection would only have one record), but rather a lottery over~them.

In ``experimental studies"~(e.g., Dwenger, K\"ubler, and Weizs\"acker, 2008 \cite{Flippingacoin}, Agranov and Ortoleva, 2017 \cite{Stochasticchoice}, Kassas, Palma, and Porter, 2022 \cite{Happytotake}, Blavatskyy, Panchenko, and
Ortmann, 2022 \cite{Howcommen}), it has been observed that participants often prefer lotteries between outcomes to choosing a certain outcome. These observations contradict the principle of stochastic dominance  in expected utility theory, which guarantees preference preservation in probabilistic mixtures\footnote{The stochastic dominance property says that for four outcomes $a,b,c$ and $d$, if $a$ is strictly preferred to $b$, and $c$ is preferred to $d$ then every probabilistic mixture between $a$ and $c$ is strictly preferred to the same probabilistic
mixture of  $b$ and $d$.}. For example, considering two  choices $a$: `Trip to Japan' and $b$: `Trip to Europe', and three options: i) receiving $a$ with probability 1, ii) receiving $b$ with probability 1, and iii) receiving $a$ with probability 0.4 and $b$ with probability 0.6, utility theory   predicts that people never prefer option iii) to options i) and ii). However,  people may prefer randomization due to their `desire for variety' (e.g., Simonson, 1990 \cite{PurchaseQuantity}, Mohan, Sivakumaran, and Sharma,  2012 \cite{Storeenvironment},  Jung and Yoon, 2012 \cite{Jung2012WhyDS}), `relief from possible regrets' (e.g., Gilovich and Medvec, 1995 \cite{ExperienceofRegret}, Zeelenberg and Pieters, 2007 \cite{RegretRegulation}), `aversion to responsibility' (e.g., Dwenger, K\"ubler, and Weizs\"acker, 
2008 \cite{Flippingacoin}), and `psychic costs' (e.g., Anderson, 2003 \cite{PsychologyofDoing}).
 
In ``matching theory" -- encompassing house allocation (e.g., Abdulkadiro{\u{g}}lu and S\"onmez, 1999~\cite{abdulkadirouglu1999house}), course assignment (e.g., Cechl{\'a}rov{\'a}, Katar{\'\i}na, Klaus, and Manlove, 2018 \cite{cechlarova2018pareto}), marriage problem (e.g., Gale and Shapley, 1962 \cite{gale1962college},  Klaus, 
2009 \cite{klaus2009fair}, Pourpouneh, Ramezanian, and Sen, 2020 \cite{pourpouneh2020marriage}), and school choice (e.g., Abdulkadiro{\u{g}}lu and S\"onmez, 2003 \cite{abdulkadirouglu2003school}, Klaus and
Klijn, 2021 \cite{klaus2021minimal},  Kojima and  \"Unver, 2014 \cite{kojima2014boston}, Pathak and S\"onmez, 2013 \cite{pathak2013school}) --   people's preferences for a set of objects are usually represented using a linear order relation and it is assumed that matching is a `one-shot' process. However, in some `frequently repeated' matching scenarios, individuals' preferences over a set of objects cannot be adequately represented by a linear order relation. Instead, their most preferred outcome may involve lotteries over the objects, which we refer to as the `ideal lottery'.

In this paper, we address the problem of ``fairly dividing chances" in repeated matching scenarios where agents express their preferences using ideal lotteries over objects that reflect their single-peaked preferences. 
To illustrate the problem of dividing chances of matching (with ideal lotteries as a  representation of single peaked preferences), consider the following scenario as a typical example: 

\noindent An IT company with four  workers labeled as $1, 2, 3$, and $4$, receives four different tasks denoted as $T_1, T_2, T_3$, and $T_4$ which arrive at the company on a recurring basis every hour.  The manager should assign tasks to the workers. Each  worker has a `desire for variety' and enjoys learning new things. Therefore, each  worker has an `ideal lottery' over tasks:
\begin{itemize}
   \centering
    \item[] $c_{1}=(T_1:0.6,T_2:0.2,T_3:0,T_4:0.2)$,
    \item[] $c_{2}=(T_1:0.5,T_2:0.1,T_3:0.3,T_4:0.1)$,
    \item[] $c_{3}=(T_1:0.6,T_2:0.1,T_3:0.1,T_4:0.2)$, 
    \item[] $c_{4}=(T_1:0.1,T_2:0.1,T_3:0.1,T_4:0.7)$.
  \end{itemize}
A worker $i$ prefers a lottery $p$ to a lottery $q$ whenever $p$ is \textit{closer}\footnote{We formally define the notion of closeness in Section~\ref{sec:Model}.} to his ideal lottery than $q$ is.  
  Workers ask the manager to fairly divide chances of doing tasks between them,  as closely as possible  to their ideal lotteries. The manager faces a fair division problem when it comes to allocating the chances.  Some tasks (like $T_1$) are in excess demand with workers and some (like $T_3$) are in excess supply (lack of demand). The manager requires  an allocation mechanism that divides the chances of doing tasks such that both \textit{agent feasibility} (for each worker, the sum of probabilities of allocating tasks to him is equal to 1) and \textit{object feasibility} (for each task, the sum of probabilities of doing the task with workers is equal to 1) are satisfied. 

\noindent Another way to explain the problem of dividing chances is by considering a scenario featuring three servers and three clients. In this scenario, each server possesses a capacity of 1 unit of computational resources, and every client demands a total of 1 unit of combined computational power from all three servers. Clients submit their ideal requests to the servers, resulting in some servers becoming overloaded while others remain underutilized. The core issue lies in efficiently balancing the computational load, requiring clients to move some of their requests from overloaded servers to underutilized ones.  The division problem revolves around determining how much load each client should shift from overloaded servers to underutilized ones.

   We can approach the ``interpretation'' of ideal  lotteries in two distinct ways: firstly, by referring to the frequency with which workers express interest in performing tasks, and secondly, by referring to the proportion of each task that captures their interest. For instance, consider   ideal lottery $c_{1}=(T_1:0.6,T_2:0.2,T_3:0,T_4:0.2)$. This ideal lottery can be understood in two different ways: firstly,  worker 1  would like to do task $T_1$ with probability $0.6$ in each hour, and secondly,  worker 1 has a preference for contributing a fraction of $0.6$ to task $T_1$. In this paper, we adopt the first interpretation. It is important to note that all the results remain applicable to the second interpretation as well. The second interpretation is adopted for task sharing between pairs of agents in Nicol\`o et al. 2023 \cite{StableSharing}.
   
``Fair division" of a resource has been a subject of substantial research across various fields, including  mathematics, economics,  and computer science (e.g., Robertson and Webb, 1998 \cite{RobBook}, Young, 1994 \cite{YoungbBook}, Moulin, 2004 \cite{MoulinBook}). This is due to the significant impact that fair resource allocation has on the design of multi agent systems, spanning various domains including   scheduling, and the allocation of computational resources such as CPU, memory, and bandwidth (Singh and Chana, 2016 \cite{CloudCom}).
 Dividing an infinitely divisible commodity among some agents with single peaked preferences is already studied (e.g., Sprumont, 1999 \cite{SPRUMONT1991Econometrica}, Thomson, 1995 \cite{Thomsonalloc1995}).  A homogeneous commodity is to be divided among a given set of agents $N$, and each agent $i$ has a single peaked preference with a peak at $x_i\in [0,1]$. If $\sum_{i\in N}x_i>1$, we say we are in excess demand, and if $\sum_{i\in N}x_i\leq 1$, we are in excess supply. One of the mechanisms to divide a commodity among the agents is the `uniform rule'.  The \textbf{Uniform Rule (UR)} (Sprumont, 1999 \cite{SPRUMONT1991Econometrica}, Thomson, 1995 \cite{Thomsonalloc1995}),  allocates to each agent his preferred share, provided it falls within common bounds which are the same for everyone and have been selected to meet the feasibility requirement.  An algorithm (see  Moulin 2004 \cite{MoulinBook}) to compute uniform rule works as follows: Start with $t=1$ and divide it into equal shares $t/n$. Let $N'$ be the set of all agents whose peaks are on the “wrong”
side of $t/n$ (if we are in excess demand, this means those agents with $x_i \leq t/n$; if we are in excess supply, it means those with $x_i \geq t/n$). Distribute the share $x_i$ to each agent  $i\in N'$ and update $t$ by subtracting $\sum_{i\in N'}x_i$ to adjust the remainder of the commodity. Update $N$ to  $N\setminus N'$ and repeat the same computation with
the remaining resources until all agents are on the same side of $t/n$,  and then give the share $t/n$ to them.
The Uniform rule is \textit{strategy proof}, meaning that agents have no incentive to misreport their preferences. It is also (Pareto) efficient, ensuring that no agent can be made better off without making some other agents worse off. This mechanism is also \textit{anonymous}, as the outcome does not depend on the agents' identities. Additionally, it is \textit{envy free}, meaning that no agent would prefer another agent's allocation over their own.
It has been demonstrated that the uniform rule is characterized by (Pareto) efficiency, strategy proofness, and anonymity (or envy freeness) properties (see Sprumont,
1999 \cite{SPRUMONT1991Econometrica}), thus it is the only mechanism that satisfies these properties.
 Division allocation and reallocation of one commodity have been studied (Thomson, 1995 \cite{Thomsonalloc1995}, Klaus,
2001 \cite{BKlausUalloc}), as well as the allocation of multiple commodities (Anno and Sasaki,
2013 \cite{Hidekazuc2013}, Morimoto,  Serizawa, and Ching, 2013 \cite{Morimoto2013c}). When dealing with multiple commodities, the uniform rule is applied independently to each commodity, and it is called \textbf{Generalized Uniform Rule~(GUR)}.

   Allocation mechanisms for divisible multiple commodities regarding multidimensional single peaked preferences are investigated in Anno and Sasaki, 2013 \cite{Hidekazuc2013}, and the generalized uniform rule (GUR) is characterized for two agents. It has been proven that, for two agents, the generalized uniform rule (GUR) is the only mechanism that satisfies `symmetry' (agents with similar preferences are treated equally), `weak peak onliness' (if an agent changes his preference without altering its peak, his outcome remains unchanged in the mechanism), and `weak second best efficiency among all strategy proof mechanisms'; this property states that the mechanism is a maximal element in the pre-ordered set of strategy proof mechanisms based on the domination relation  (see Anno and Sasaki, 2013 \cite{Hidekazuc2013}).   Also, in the context of multiple commodities, for continuous, strictly convex, and separable multidimensional single peaked  preferences, it has been proven that a mechanism satisfies strategy proofness, unanimity (if the total sum of peak values for each commodity matches the supply of that commodity, then each agent should be assigned according to their respective peaks), weak symmetry, and non-bossiness (when an agent’s preferences change, if his allocation remains the
same, then the whole outcome of the mechanism   should remain the same) if and only if it is GUR (see Morimoto,  Serizawa, and Ching,~ 2013~\cite{Morimoto2013c}).

However, we cannot apply the general uniform rule (GUR) for dividing chances. 
   In the division of multiple commodities, preferences are represented using separable multidimensional single peaked preferences, whereas, in the \textbf{division of chances}, lotteries are treated as whole entities with the sum of their probabilities equal to 1, which makes it \textbf{impossible} to utilize separable multidimensional single peaked preferences,  as the chances of receiving objects are not separable. Also because of \textit{agent feasibility} for dividing chances, we cannot apply GUR (see Example~\ref{exm:divisionchanc}). 
Since the General Uniform Rule (GUR) is not suitable for addressing the problem of dividing chances in one-to-one matching scenarios, we propose   a new class of mechanisms, referred to as Uniform Rule for Dividing Chances mechanisms (\textbf{URC mechanisms}).
In informal terms,  a URC mechanism operates as follows:

\noindent We have a set of agents, denoted as $N$, and a set of objects, denoted as $A$, where the number of agents is equal to the number of objects. Each agent possesses an ideal lottery over the objects. Think of each object $a \in A$ as a tank containing one liter of a specific colored liquid named $a$. Each agent $i\in N$ is equipped with a tap, $tap_{ia}$, for each tank $a\in A$. Additionally, each agent $i$ has a bucket $buk_i$ with a capacity of one liter, and the ideal lottery of agent $i$, denoted by $c_i=(c_{ia})_{a\in A}$, signifies the desired combination of liquids that agent $i$ wishes to have in their bucket.
The mechanism functions in two phases. In phase 1, all agents simultaneously open their taps on the tanks. Each agent $i$ closes tap $tap_{ia}$ as soon as the amount of liquid $a$ in their bucket reaches their ideal for object $a$. Phase 1 concludes when one of two conditions is met: either all taps have been closed, or, if any tap on a tank remains open, it signifies that the liquid in that tank has been fully exhausted.

 \begin{figure}[h!]
 \centering
 \includegraphics[scale=0.8]{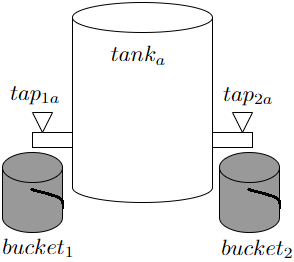} 
  \caption{URC Mechanisms}
 
  \label{fig:tank}
\end{figure}

\noindent At the end of phase 1, it is possible that some tanks are not entirely empty and some buckets are not yet completely filled. So, in phase 2, we organize all the buckets in a line, and separately, we arrange all the tanks in another line.
In this phase, we initiate the process by opening the tap of the first tank, allowing the liquid to flow into the first bucket. If the tank becomes empty, we proceed to the next tank, and if the bucket becomes full, we move on to the next bucket. By the end of phase 2, each bucket contains exactly one liter of liquid. The chance of allocating object $a$ to agent $i$ is proportional to the amount of liquid of object $a$ present in the corresponding bucket $buk_i$.  
Different ordering of tanks and buckets defines different URC mechanisms. We establish that all URC mechanisms are equivalent in terms of welfare, meaning that for each agent, the distance between  agent's allocation and their ideal lottery remains the same regardless of the ordering chosen for the tanks and buckets.

We discuss ``properties" of URC mechanisms in our context. We show that URC mechanisms are \textit{strategy proof} and \textit{efficient}.
We introduce the concept of \textit{replacement monotonicity} for dividing chances, and show that URC mechanisms satisfy this property. This concept was originally introduced in Barber\`a, Jackson, and
Neme 1997 \cite{Jackson1997}, to characterize sequential allotment mechanisms for dividing a commodity. Roughly, the replacement monotonicity property can be understood as follows: If an agent alters their preferences in a manner that frees up excess demand resources (objects) without improving their own welfare, it will not negatively impact the welfare of other agents.
We introduce the concept of \textit{in-betweenness} for dividing chances, and show that URC mechanisms satisfy this property. In-betweenness is a type of monotonicity property that examines situations where if an agent increases their ideals for objects that the mechanism allocated them more than their ideal, and decreases their ideals for objects that the mechanism allocated them less than their ideal, then this change does not lead to a decrease for objects they already received more and does not lead to an increase for objects they already received less.
We also show that URC mechanisms are anonymous, envy free, and non-bossy.  In  one-shot matching scenarios where preferences are represented with linear order relations, achieving strategy proofness, efficiency, and fairness simultaneously is impossible (e.g., Bogomolnaia and  Moulin, 2001 \cite{BM2001}, Nesterov 2017 \cite{NES2017}, Ramezanian and Feizi, 2022 \cite{Ram2022}).  However,   in repeated matching scenarios with ideal lotteries as preferences, it becomes possible to satisfy these axioms (Proposition~\ref{Theorem:URL-properites}).

``Characterization" theorems help us understand the inner workings of mechanisms  used to allocate resources. By revealing the essential properties that define these mechanisms, we gain insights into how they function and make decisions. 
We characterize the class of URC mechanisms, up to welfare equivalence, based on  key properties:  strategy proofness, (Pareto) efficiency, replacement monotonicity,  non-bossiness, in-betweenness, and anonymity (envy freeness). By `characterization up to welfare equivalence'\footnote{In Section~\ref{sec:Model}, we provide a formal explanation of what we mean by 'characterization up to welfare equivalence'.}, we mean that when a mechanism satisfies all these properties, it is equivalent in terms of welfare (distance between the agents' allocations and their ideal lotteries) to URC~mechanisms.

Investigating ``logical independency" of properties aims to shed light on the relationships between properties and to better understand the trade-offs and interactions within allocation mechanisms.  We explore the independence of properties that characterize (up to welfare equivalence) URC mechanisms by introducing alternative mechanisms that fulfill certain properties while lacking others. 

 The paper is organized as follows: In Section~\ref{sec:Model}, we formally model the  division problem of chances and discuss properties such as (Pareto) efficiency and replacement monotonicity  within this framework.  
In Section~\ref{sec:Mechan}, we introduce a  class of mechanisms, referred to as Uniform Rule for Dividing Chances  (URC mechanisms)   in one-to-one matching scenarios,  and prove that URC mechanisms satisfy strategy proofness, (Pareto) efficiency, replacement monotonicity, non-bossiness, in-betweenness, envy freeness, and anonymity.
In Section~\ref{sec:charact}, we characterize the URC mechanisms  up to welfare equivalency.
In Section~\ref{sec:indaxiom}, we study the logical relationship between properties: strategy proofness, efficiency, replacement monotonicity,  in-betweenness,  non-bossiness,  anonymity and envy freeness.
Section~\ref{sec:conclu} concludes the paper and discusses further work. The proofs of the lemmas, propositions, and theorems are provided in Appendix~\ref{sec:appen}.

\section{Model} \label{sec:Model}

To define single peaked preferences for lotteries, we first need to establish a notion of distance, as the preference is determined by the closeness to the peak. For a finite set of objects, $A$, we define $\Delta(A)$ as the set of all lotteries over $A$. We also define $\mathbf{d}:\Delta(A)\times \Delta(A)\xrightarrow{} \mathrm{R}$ as the metric function  $\mathbf{d}(p,q)=\sum_{a\in A}\lvert p_a-q_a\rvert$, for $p,q\in \Delta(A)$. 
 While it is possible to consider alternative distance metrics,  such as $d(p,q)=(\sum_{a\in A}(\lvert p_a-q_a\rvert)^\rho)^{1/\rho}$, for $\rho\geq 1$ (referred to as $\mathit{l}^\rho$ distances),  in this paper, we specifically focus on the $\mathit{l}^1$ distance, and all of our results  hold regarding this distance. This choice aligns with the utilization of the $\mathit{l}^1$ distance for single peaked preferences within the literature of voting and aggregating mechanisms of \textit{probabilistic distributions} (see Freeman et al.~2021, \cite{TruthAgg}, and Goel~et al.~2019 \cite{VotAgg}). 

The \textbf{Division  Problem of Chances}   is defined as a tuple $(N,A,(c_i)_{i\in N})$ where $N$ is a set of agents, $A$ is a set of objects  with $|A|=|N|=n$, and $(c_i)_{i\in N}$ is a \textit{preference profile} where each agent $i\in N$ has an ideal lottery $c_i$ over objects such that for each $a\in A$,   $c_{ia}\geq 0$ represents the desired chance (fraction) of receiving object $a$ for agent $i$, and $\sum_{a\in A}c_{ia}=1$. We refer to the set of all preference profiles  by $C^n$. 

\noindent Welfare of an agent increases by getting close to its ideal lottery.  An agent $i$ \textbf{prefers a lottery $p$ to a lottery $q$}, denoted by $p \succ_{c_i} q$, whenever $\mathbf{d}(c_i,p)<\mathbf{d}(c_i,q)$. The notations $p \succeq_{c_i} q$ and $p \equiv_{c_i} q$ respectively refer to $\mathbf{d}(c_i, p) \leq \mathbf{d}(c_i, q)$ and $\mathbf{d}(c_i, p) = \mathbf{d}(c_i, q)$. 
Note that since $\mathbf{d}$ is a metric, $\mathbf{d}(c_i,p)=\mathbf{d}(p,c_i)$, and we use both of them interchangeably in the following sections. 


Given a profile of preferences $c$, we say an object $a\in A$ is  \textit{unanimous} whenever $\sum_{i\in N}c_{ia}=1$, in    \textit{excess demand} whenever $\sum_{i\in N}c_{ia}>1$, and    in \textit{excess supply} whenever $\sum_{i\in N}c_{ia}<1$.   We use $ED(c)$ to denote the set of excess demand objects, $ES(c)$ for the set of excess supply objects, and $UN(c)$ for unanimous ones. 
A \textit{random matching} is  bistochastic matrix  $P=[p_{ia}]_{i\in N, a\in A}$,  where for all $i\in N, a\in A$, $p_{ia}\geq 0$, and
 
\begin{itemize}
    \item for each object $a\in A$,  $\sum_{i\in N}p_{ia}=1$ (\textit{object feasibility}) and
    \item for each agent $i\in N$, $\sum_{a\in A}p_{ia}=1$ (\textit{agent feasibility}). 
\end{itemize}
  For a random matching  $P=[p_{ia}]_{i\in N, a\in A}$, for every agent  $i\in N$, $p_i=(p_{ia})_{a\in A}$ is called the \textit{allocation} of agent $i$. And, for every object $a\in A$, $p_a=(p_{ia})_{i\in N}$ is the division of object $a$.
     For every profile $c=(c_i)_{i\in N}$, for every $i\in N$, we let $c_{-i}=(c_1,c_2,...,c_{i-1},c_{i+1},...,c_n)$. Also  for $i\in N$, for $a\in A$, we use notations $c_i=(c_{ia})_{a\in A}$, and  $c_a=(c_{ia})_{i\in N}$.

\noindent Let $l_1$ and $l_2$ be two lotteries over a set of objects $A$. We say a lottery $l$ is \textit{between} $l_1$ and $l_2$, denoted by $ l\in Between(l_1,l_2)$ whenever for every $a\in A$, either $l_1(a)\leq l(a)\leq l_2(a)$ or  $l_2(a)\leq l(a)\leq l_1(a)$.

A mechanism $\mu$ is a mapping from the set of preference profiles to the set of random matchings.
Two mechanism $\mu$ and $\mu'$  are considered \textit{welfare equivalent} when,  for every profile $c=(c_i)_{i\in N}$, for every agent $i\in N$, it holds that $\mu_i(c)\equiv_{c_i} \mu'_i(c)$.  

Let $\Phi_1,\Phi_2,...,\Phi_k$ be some properties for mechanisms and $\mathcal{M}$ be a class of mechanisms.  We say $\Phi_1,\Phi_2,...,\Phi_k$ \textit{characterize $\mathcal{M}$ up to welfare equivalence}\footnote{It's worth noting that if the notion of equivalence is replaced by equality, we obtain the well-known notion of characterization.} whenever 
\begin{itemize}
    \item[\textit{i-}] every $\mu\in \mathcal{M}$ satisfies $\Phi_1,\Phi_2,...,\Phi_k$, and
    \item[\textit{ii-}] for every mechanism $\mu'$, if $\mu'$ satisfies $\Phi_1,\Phi_2,...,\Phi_k$ then $\mu'$ is welfare equivalent to some $\mu\in \mathcal{M}$.
\end{itemize} 
Note that $\Phi_1, \Phi_2, ..., \Phi_k$ are properties of mechanisms and not of welfare equivalence classes. Therefore, it is possible for two mechanisms to be welfare equivalent but not satisfy similar properties~(see Figure~\ref{fig:WEL}).

 A mechanism is strategy proof if it is in the best interest of each participant to truthfully reveal their preferences, as misrepresentation does not improve their welfare.  

\begin{itemize}
    \item[] \textbf{Strategy Proofness}: A mechanism $\mu$ is \textit{strategy proof} whenever   for every  profile $c=(c_i)_{i\in N}$, for every agent $i\in N$, there exists no ideal lottery $c'_i$, such that $\mu_i(c'_i,c_{-i}))\succ_{c_i} \mu_i(c)$. 
\end{itemize}

\subsection{Efficiency}
We discuss the concept of   efficiency and same-sideness in our context, and prove that these two notions  are logically equivalent.  Given a profile of preferences $c=(c_i)_{i\in N}$,  A random matching $P$ is considered (Pareto) efficient if it is impossible to enhance the welfare of certain agents without diminishing the welfare of others.
More formally, a random matching $P$ is  \textbf{lottery dominated} by a random matching $Q$ whenever   for all agent $i\in N$, $q_i\succeq_{c_i} p_i$,
and     \textbf{strictly lottery dominated} by $Q$ whenever it is lottery dominated and also there exist some $j\in N$ such that $q_j\succ_{c_j} p_j$.
A random matching $P$ is (Pareto) \textbf{efficient} if it is not  strictly lottery dominated by any other random matching. For convenience, we use the term `efficiency' instead of `Pareto efficiency'.

A random matching $P$ is \textbf{same-sided} whenever for every object $a\in A$, if $a$ is in \textit{excess demand} then for all $i\in N$, $p_{ia}\leq c_{ia}$, if $a$ is in \textit{excess} supply then for all $i\in N$, $p_{ia}\geq c_{ia}$, and if $a$ is  \textit{unanimous} then for all $i\in N$, $p_{ia}=c_{ia}$.

  The following lemma demonstrates that, given a same-sided random matching $P$ with respect to a profile $c$, the distance between an agent $i$'s allocation ($p_i$) and his ideal lottery ($c_i$) how is computed. Note that the Lemma~\ref{lem:distance-sameside}, works only for same-sided random matchings. 

\begin{lemma}\label{lem:distance-sameside}
    Given a profile of preferences $c$, if $P$ is same-sided then for all $i\in N$,
    \begin{equation*}
      \sum_{a\in ED(c)}(c_{ia}-p_{ia})=\sum_{a\in ES(c)}(p_{ia}-c_{ia}),~ and   
    \end{equation*}

\begin{equation}\label{distance-sameside}
   \mathbf{d}(p_i,c_i)=2\times \sum_{a\in ED(c)}(c_{ia}-p_{ia})= 2\times \sum_{a\in ES(c)}(p_{ia}-c_{ia}). 
\end{equation}

\end{lemma} 
\begin{proof} 
    The proof is straightforward as follows: 
    \begin{center}
        $0=1-1=\sum_{a\in A}c_{ia} - \sum_{a\in A}p_{ia}= (\sum_{a\in ED(c)}(c_{ia}-p_{ia}))+ (\sum_{a\in ES(c)}(c_{ia}-p_{ia}))+ (\sum_{a\in UN(c)}(c_{ia}-p_{ia}))$,
    \end{center} thus $\sum_{a\in ED(c)}(c_{ia}-p_{ia})=\sum_{a\in ES(c)}(p_{ia}-c_{ia}) $, and

\begin{center}
    $\mathbf{d}(p_i,c_i)=\sum_{a\in ED(c)}(c_{ia}-p_{ia})+\sum_{a\in ES(c)}(p_{ia}-c_{ia})= 2\times \sum_{a\in ED(c)}(c_{ia}-p_{ia}) $.
\end{center}
  \end{proof}
  
\noindent  In Proposition~\ref{eff-same}, we demonstrate that efficiency and same-sideness are logically equivalent.

\begin{proposition}\label{eff-same}  A random  matching $P$ is efficient if and only if it is same-sided. 
\end{proposition}
\begin{proof}
   See Appendix~\ref{app:Model}. 
\end{proof}


\subsection{Replacement Monotonicity} 
The concept of replacement monotonicity can be classified as a kind of symmetry property. It was originally introduced in Barber\`a, Jackson, and Neme, 1997 \cite{Jackson1997}, to characterize sequential allotment mechanisms for dividing a commodity. It is also used in Masso and
Neme, 2007 \cite{Bribe2}, to characterize bribe proof mechanisms. 
The concept of replacement monotonicity is relevant when an individual's preferences change, potentially affecting their allocation. In such cases, there should be an offsetting change in the allocations of the other individuals. Replacement monotonicity demands that the allocations of the remaining individuals do not move in opposite directions. More formally,
For dividing a  commodity between agents in $N$, suppose that each agent $i\in N$ has a single peaked preference with a peak at $x_i$. A mechanism $\phi$ is replacement monotonic whenever 
\begin{center}
   if $\phi_i(x_i,x_{-i})\leq \phi_i(x'_i,x_{-i})$ then for all $j\neq i$, $\phi_j(x_i,x_{-i})\geq \phi_j(x'_i,x_{-i})$.
\end{center}
Replacement monotonicity asserts that if the allocation of some agent $i$  does not decrease due to the change from  $x_i$ to $x'_i$,  then
the allocations of all other agents do not increase.

Drawing inspiration from the concept of replacement monotonicity  presented in Barber\`a, Jackson, and Neme, 1997 \cite{Jackson1997}, we revise the property of replacement monotonicity  for welfare of agents in dividing chances: When an individual alters their preferences, which can potentially impact their welfare, there should be a compensatory adjustment in the welfare of the remaining individuals. 
Let $\mu$ be a  mechanism, and  $c=(c_i)_{i\in N}$ a profile of preferences. If one agent $i$ (say a worker) modifies his ideal lottery in a way that
\begin{itemize}
    \item his modification is with the aim to free up the demand to   excess demand objects (tasks), formally,~for~all~$a\in ED(c)$, $c'_{ia}\leq c_{ia}$,
    \item his  modification does not result in another object (task), which was not already in excess demand, becoming in excess demand. More precisely,  $ED((c'_i,c_{-i}))= ED(c)$), and 
    \item the agent's own welfare does not increase as a result of this modification,
\end{itemize} 
 then there is no decrease in the welfare of  other agents.

 In other words, when an agent $i$ changes their preferences from $c_i$ to $c'_i$, freeing up some excess demand objects,  this change does not lead to make some excess supply or unanimous  objects becoming in excess demand, and additionally, agent $i$'s own welfare does not increase as a result of this change, then the welfare of all other agents does not decrease either.
 \begin{definition}\label{DEF:RM}
     A mechanism $\mu$ is replacement monotonic whenever for every  profile $c=(c_j)_{j\in N}$, for every agent $i\in N$, for every   lottery $c'_i$ that satisfies 
     \begin{itemize}
         \item $ED((c'_i,c_{-i}))= ED(c)$, and
         \item for all~$a\in ED(c)$,~$c'_{ia}\leq c_{ia}$, 
     \end{itemize}  
     \begin{center}
       if $\mu_i(c)\succeq_{c_i} \mu_i((c'_i,c_{-i}))$   then for all $j\in N\setminus\{i\}$,   $\mu_j((c'_i,c_{-i}))\succeq_{c_j} \mu_j(c)$ .   
     \end{center}
    
 \end{definition}

\subsection{Non-Bossiness}\label{nonbos}
We can define the concept of non-bossiness with two approaches: one with welfare and one based on allocations. A mechanism is considered welfare non-bossy when a modification in an agent's preferences that leaves their own welfare unchanged also ensures that no one else's welfare is affected. Formally, a mechanism $\mu$ is welfare non-bossy whenever for each $c\in C^n$, each $i\in N$, and each $c'_i\in C$, if 
$\mu_i(c)\equiv_{c_i} \mu_i((c'_i,c_{-i}))$
then for all $j\in N\setminus\{i\}$,
$\mu_j(c)\equiv_{c_j} \mu_j((c'_i,c_{-i}))$.
However, the URC mechanisms, the subject of study in this paper, do not satisfy welfare non-bossiness (see Example~\ref{Exam:non}). 

 We consider the concept of non-bossiness based on allocations for excess demand objects. A mechanism is  non-bossy when a modification in an agent's preferences that leaves their own allocation for excess demand objects unchanged also ensures that no one else's allocation for excess demand objects is affected.

   \begin{itemize}
     \item \textbf{Non-bossiness}: For each $c\in C^n$, each $i\in N$, and each $c'_i$, if  for all $a\in ED(c)$, $\mu_{ia}((c'_i,c_{-i}))=\mu_{ia}(c)$.    then for all $j\in N\setminus\{i\}$, for all $a\in ED(c)$, $\mu_{ja}((c'_i,c_{-i}))=\mu_{ja}(c)$.
   \end{itemize} 
So, our concept of non-bossiness considers only allocations for excess demand objects and not all objects. Also note that if $\mu$ is a same-sided mechanism, by Lemma~\ref{distance-sameside}, then starting from the sentence `for all $j\in N\setminus\{i\}$, for all $a\in ED(c)$, $\mu_{ja}((c'_i,c_{-i}))=\mu_{ja}(c)$' we can conclude the sentence `for all $j\in N\setminus\{i\}$,
$\mu_j(c)\equiv_{c_j} \mu_j((c'_i,c_{-i}))$'.
    
 \subsection{Envy Freenss, Anonymity, In-Betweeness}

In this section, we establish the concepts of envy freeness, anonymity and in-betweenness for the division problem of chances.  A random matching $P$ is envy free when no individual prefers the allocation of another over their own, meaning no one is envious of others' outcomes.
\begin{itemize}
    \item\textbf{Envy Freeness}: For each $c\in C^n$, for every $i,j \in N$,  $p_i\succ_{c_i} p_j$. 
\end{itemize}

In our context, anonymity entails that the welfare of the final outcome under a mechanism is independent of the specific identities of the participating agents.
\begin{itemize}
   \item\textbf{Anonymity}: Let $H$ be a permutation of $N$, $c,c'\in C^n$, and for every $i\in N$, $c'_{H(i)}=c_i$. We have   $\mu_i(c)\equiv_{c_i} \mu_{H(i)}(c')$.
    
\end{itemize}
 

The In-Betweenness property is a type of monotonicity properties (Thomson, 2011 \cite{THOMSON2011393}, Section 7), where a mechanism allocates monotonically with respect to varying peaks of objects in a specified set. For a profile $c$, for an agent $i$, and $c'_i$ between $c_i$ and $\mu_i(c)$, we say    an object $a$ is a shortage object for $i$, whenever $c_{ia}\geq \mu_{ia}(c)$, and   an abundance object $a$ for $i$, whenever $c_{ia}\leq \mu_{ia}(c)$.
In-betweenness property asserts as follows:  If an  agent $i$ (a worker) decreases his ideal  for a shortage object (task) $a$ from $c_{ia}$ to $c'_{ia}$, but not less than what the mechanism is allocated to him — that is,  $\mu_{ia}(c)\leq c'_{ia} \leq c_{ia}$ — then his allocation   does not increase, that is $\mu_{ia}((c'_i,c_{-i}))\leq\mu_{ia}(c)$. Similarly, if an agent $i$ (a worker) increases his ideal for an abundance object (task) $a$, but not more than what the mechanism is allocated to him — that is,  $\mu_{ia}(c)\geq c'_{ia} \geq c_{ia}$ — then his allocation   does not decrease, that is $\mu_{ia}((c'_i,c_{-i}))\geq\mu_{ia}(c)$.

\begin{itemize}
    \item \textbf{In-Betweenness}: For each $c\in C^n$, for every $i\in N$, if $c'_i$ is between $c_i$ and $\mu_i(c)$, then    for~all~$a\in A$:
    
    if $c'_{ia}\leq c_{ia}$ we have $\mu_{ia}((c'_i,c_{-i}))\leq \mu_{ia}(c)$, and if $c'_{ia}\geq c_{ia}$ we have  $ \mu_{ia}((c'_i,c_{-i}))\geq \mu_{ia}(c)$.
    \end{itemize}
    In other words, a mechanism satisfies in-betweenness whenever it is monotonic  for lotteries in $between(c_i, \mu_i(c))$.
We conclude this section with the ensuing proposition and its corollary.
    \begin{proposition}\label{prop:Between}
        Suppose $\mu$ is a strategy proof, efficient, and in-between mechanism. Let $c=(c_i)_{i\in N}$ be a profile. If an agent $i$ modifies his ideal lottery $c_i$ to a lottery $c'_i$ where $c'_i$ is between $c_i$ and $\mu_i(c)$, then for all $a\in A$, $\mu_{ia}(c)=\mu_{ia}((c'_i,c_{-i}))$.
    \end{proposition}
    \begin{proof}
        See Appendix~\ref{app:Model}.
    \end{proof}

    \begin{corollary}\label{cor:between}
        Suppose mechanism $\mu$ is  strategy proof, efficient, in-between and non-bossy. Let $c=(c_i)_{i\in N}$ be a profile. If an agent $j$ modifies his ideal lottery $c_j$ to a lottery $c'_j$ where $c'_j$ is between $c_j$ and $\mu_j(c)$, then 
for all $a\in ED(c)$, for all $i\in N$, $\mu_{ia}(c)=\mu_{ia}((c'_j,c_{-j}))$.  
    \end{corollary}
      



  In the next section, we discuss URC mechanisms (Uniform Rules for dividing Chances) and prove that every URC mechanism is strategy proof, efficient, replacement monotonic, non-bossy,  in-between, anonymous, and envy free.  

 \section{Uniform Rules for Dividing Chances} \label{sec:Mechan}
  In the introduction and Example~\ref{exm:divisionchanc}, we discussed the difference between dividing chances in matching scenarios and dividing quantities of multiple commodities. This difference arises from  the   requirement of ``agent feasibility" and the fact that lotteries are
treated as whole entities that makes the application of the generalized uniform rule, GRU, inappropriate for dividing chances.
In this section, we introduce a class of mechanisms for dividing chances, called Uniform Rule for dividing Chances (URC) mechanisms, and investigate their properties. 

We begin by revisiting the uniform rule (Sprumont 1999 \cite{SPRUMONT1991Econometrica}), denoted as UR, for dividing a homogeneous commodity among a given set of agents  $N$, where each agent $i$ has a single peaked preference with a peak at $x_i\in [0,1]$.  The uniform rule   allocates to each agent his preferred share, provided it falls within common bounds which are the same for everyone and have been selected to meet the feasibility requirement.   
Given a profile $x=(x_i)_{i\in N}$, where $x_i$ is the peak of agent $i$, the mathematical function representing  the uniform rule is  as follows: $UR: [0,1]^n\xrightarrow{} [0,1]^n$,   \[  UR_i(x) = 
   \begin{cases}
     \min(x_i,\lambda(x))  &\quad\text{if}~ \sum_{j=1...n}x_j\geq 1 \\
     \max(x_i,\nu(x))  &\quad\text{if}~ \sum_{j=1...n}x_j\leq 1 \\ 
    \end{cases}
 \]  where $\lambda(x)$ solves the equation 
$\sum_{i\in N} \min(x_i,\lambda(x))=1$ and
$\nu(x)$ solves   $\sum_{i\in N} \max(x_i,\nu(x))=1$.

Now, we are prepared to introduce the class of URC mechanisms. Think of each object $a \in A$ as a tank containing one liter of a specific colored liquid named $a$. Each agent $i\in N$ is equipped with a tap, $tap_{ia}$, for each tank $a\in A$. Additionally, each agent $i$ has a bucket $buk_i$ with a capacity of one liter. The ideal lottery of agent $i$, denoted by $c_i=(c_{ia})_{a\in A}$, signifies the desired combination of liquids that agent $i$ wishes to have in their bucket.
The mechanism functions in two phases. In phase 1,

\begin{itemize}
    \item all agents simultaneously open their taps on the tanks. Each agent $i$ closes tap $tap_{ia}$ as soon as the amount of liquid $a$ in their bucket reaches their ideal for object $a$, that is $c_{ia}$.
    \item Phase 1 concludes when one of two conditions is met: either all taps have been closed, or, if any tap on a tank remains open, it signifies that the liquid in that tank has been fully exhausted.
\end{itemize} 
Assuming $w_{ia}$ is the current amount of liquid $a$ in bucket $i$. For objects that are not in excess demand, agent $i$ closes their tap at their ideal points, so $w_{ia}=c_{ia}$. For objects that are in excess demand, either agent $i$ closes their tap before liquid $a$ is exhausted, or they, along with all other agents who cannot reach their ideal points, receive an equal amount of $a$ while keeping their taps open until all liquid $a$ is exhausted, this process implies that excess demand objects are divided based on the uniform rule. Therefore, 
\begin{itemize}
    \item \textbf{Phase 1}:  For every object $a\in A$, 
    \begin{itemize}
        \item if $a\not\in ED(c)$, for all $i\in N$, let $w_{ia}=c_{ia}$,
         \item if $a\in ED(c)$, for all $i\in N$, let $w_{ia}=UR_i(c_{a})$.
    \end{itemize}
\end{itemize}
\noindent At the end of phase 1, it is possible that some tanks (objects) are not entirely empty (exhausted) and some buckets (agents) are not yet completely filled. Indeed, all objects except excess supply objects are exhausted in phase 1, so, after the completion of phase 1, excess supply object are not be exhausted, and some agents do not have reached their capacity (where the sum of the probabilities they hold is less than 1). Therefore, phase 2 of the mechanism is executed as follows to ensure
object feasibility and agent feasibility.

In phase 2, we organize all the buckets in a line, and separately, we arrange all the tanks in another line. We initiate the process by opening the tap of the first tank, allowing the liquid to flow into the first bucket. If the tank becomes empty, we proceed to the next tank, and if the bucket becomes full, we move on to the next bucket. By the end of phase 2, each bucket contains exactly one liter of liquid. The chance of allocating object $a$ to agent $i$ is proportional to the amount of liquid of object $a$ present in the corresponding bucket $buk_i$.

More formally Let $\alpha$ represent a sequence of all agents in $N$, and $\beta$ a sequence of all objects in $A$. The second phase operates based on following informal pseudocode:
\begin{itemize}   
    \item \textbf{Phase 2}:  
\begin{itemize}
      \item[0.] Initialize $t=1, s=1$  
    \item[1.] If the bucket of agent $\alpha_t$ is full (i.e., $\sum_{a\in A}w_{\alpha_ta}=1$) then proceed to the next agent in the sequence by updating $t$ to  $t+1$, continuing updating $t$ to  $t+1$ until either finding an agent $\alpha_t$ in the sequence that their bucket is not full, or  the sequence $\alpha$ is concluded ($t=n$).
    \item[2.] If  object $\beta_s$ is exhausted (i.e., $\sum_{i\in N}w_{i\beta_s}=1)$ then   move on to the next object by updating $s$ to  $s+1$, continuing  updating $s$ to  $s+1$  until either finding an object $\beta_s$ that the liquid in its tank is not exhausted, or the sequence $\beta$ is concluded ($s=n$).
    \item[3.] Update $w_{\alpha_t\beta_s}$ to  $w_{\alpha_t\beta_s}+\min\bigg((1-\sum_{a\in A}w_{\alpha_ta}),(1-\sum_{i\in N}w_{i\beta_s})\bigg)$.
    
    The value $(1-\sum_{a\in A}w_{\alpha_ta})$ is the current free capacity of the bucket of agent $\alpha_t$, and $(1-\sum_{i\in N}w_{i\beta_s})$  is the current remaining liquid in the tank of object $\beta_s$. 
    \item[4.] If $t=n$ and $s=n$, Stop, otherwise return to Step 1.

\end{itemize}
\end{itemize}
Finally, we let for every $i\in N$, for every $a\in A$, $URC^{\alpha,\beta}_{ia}(c)=w_{ia}$ as the outcome of the mechanism.

 Phase 2 of URC mechanisms solely forces agents to choose additional fractions of objects to ensure agent feasibility as well as object feasibility, and does not impact welfare. So, all URC mechanisms are equivalent in terms of welfare.
 
 \begin{proposition}\label{prop:seqInd}
     For every sequences $\alpha$, $\alpha'$, $\beta$, and $\beta'$, for every profile $c$, for all $i\in N$,
 \begin{center}
     $URC^{\alpha,\beta}_{i}(c)\equiv_{c_i} URC^{\alpha',\beta'}_{i}(c)$.
 \end{center}
 \end{proposition}
 \begin{proof}  
The outcome of URC mechanisms is efficient, as it is shown in Proposition~\ref{Theorem:URL-properites}. According to Proposition~\ref{eff-same}, efficiency is equivalent to same-sideness. So, with the help of Lemma~\ref{lem:distance-sameside}, Equation~(\ref{distance-sameside}), we can compute the distance based on excess demand objects.
In phase 1, all excess demand objects are exhausted, and therefore, the allocation of excess demand objects is determined in phase 1. Phase 1 is independent of sequences $\alpha$ and $\beta$; thus welfare is independent of sequences $\alpha$ and $\beta$, and we are done.
       \end{proof}
 
\begin{notation}
For convenience, we use $URC$ instead of $URC^{\alpha,\beta}$, when   sequences $\alpha$ and $\beta$ are arbitrary.
\end{notation}

\begin{example}\label{exam:1}

Suppose $N=\{1,2,3\}$, $A=\{a,b,c\}$ and the preference profile $c$ is
\begin{equation}\label{eq:exam1}
    \begin{matrix}
c_1=(a:0.6,b:0.2,c:0.2)\\
c_2=(a:0.5,b:0.4,c:0.1)\\
c_3=(a:0.2,b:0,c:0.8)~~\\
\end{matrix}
\end{equation}
   
\noindent    For   profile $c$, we have $ED(c)=\{a,c\}$. The execution of phase 1 leads to the following outcomes: 
\begin{itemize}
    \item $w_{1a}=0.4$, $w_{2a}=0.4$ and $w_{3a}=0.2$.
    \item $w_{1c}=0.2$,  $w_{2c}=0.1$,  and $w_{3c}=0.7$.

    \item  $w_{1b}=0.2$,  $w_{2b}=0.4$,  and $w_{3b}=0$.
\end{itemize} 
  
Let $\alpha=123$, and $\beta=abc$.
In Phase 2, the remaining chances of object $b$ is assigned to all agents, making them full. The final random matching is:
\begin{itemize}
    \item $(p_{1a}=0.4,p_{1b}=0.4,p_{1c}=0.2)$,
      \item $(p_{2a}=0.4,p_{2b}=0.5,p_{2c}=0.1)$,  
      \item $(p_{3a}=0.2,p_{3b}=0.1,p_{3c}=0.7)$.
\end{itemize}
    \end{example}

 We show that URC mechanisms satisfy efficiency, strategy proofness,   replacement monotonicity, non-bossiness, envy freeness, and in-between.

\begin{proposition}\label{Theorem:URL-properites} URC mechanisms are    efficient, strategy proof, in-between,      non-bossy,  replacement monotonic, anonymous, and envy free.
\end{proposition}
\begin{proof}
    See Appendix~\ref{app:Mechan}.
\end{proof}

\subsection{Exploring Alternative Mechanisms for Dividing Chances}\label{subsec:explor}

In this section, we explore the possibility of constructing alternative mechanisms for dividing chance. Then, in Section~\ref{sec:charact}, we demonstrate that any mechanism satisfying the desired properties stated in Theorem~\ref{charac} is welfare equivalent to URC mechanisms.

Let's approach the division problem of chances as a constraint satisfaction problem and explore the role of a mechanism designer. The mechanism designer is presented with a profile of preferences, denoted as $c=(c_i)_{i\in N}$, faced with the problem that, in $c=(c_i)_{i\in N}$, some objects are in excess demand whereas others are in excess supply, and the designer encounters difficulty in satisfying all agent preferences directly as he needs to satisfy object feasibility as a constraint. To address this, the designer opts to transfer requests from excess demand objects to excess supply objects to ensure object feasibility. The key question becomes \textit{determining the appropriate amount of transfer} for each agent. 

In URC mechanisms, to ascertain the  appropriate amount of transfer, we concentrate on the surplus of objects in excess demand. Utilizing the uniform rule, we determine, for each agent, how much of each excess demand object should be transferred to excess supply objects. In following, we discuss two other alternative approaches two determine the appropriate amount of transfer.

\noindent 1)  One way to determine the appropriate amount to transfer is to focus on the deficiency for objects which are in excess supply. Consider the following profile for three agents and three objects:
      \begin{itemize}
      \centering
        \item[] $c_1=(a:0,b:1,c:0)$,
        \item[] $c_2=(a:0,b:1,c:0)$,
        \item[] $c_3=(a:0,b:0,c:1)$.
        
    \end{itemize}
For object $a$, we are in excess supply. The mechanism designer may decide to determine how much each agent should have from object $a$, and then transfer that amount from other objects to object~$a$.

 If the mechanism designer, employing the uniform rule (UR), aims to distribute object $a$ among agents, then we have $p_{1a}=p_{2a}=p_{3a}=1/3$. Due to same-sideness, the mechanism designer cannot transfer anything from object $c$. So, they should transfer the required amount from object~$b$. However, agent $3$ has nothing from object $b$, and transferring an amount from object $b$ results in a negative quantity of $b$ for agent $3$. 
 
 The same argument holds for any other mechanisms such as the proportional rule, that allocates a positive amount of object $a$ to agent $3$, and leads to a negative quantity of object $b$ for agent $3$, making it an impractical choice. We conclude that there is no mechanism that can determine the appropriate amount to transfer fairly based solely on the ideals of agents for excess supply objects. 

We say two profiles $c$ and $c'$ coincide on a subset $S\subseteq A$ when, for every agent $i\in N$, and every object $a\in S$, $c_{ia}=c'_{ia}$.
In URC mechanisms, the welfare of all agents (the distance between their allocation and their ideal lottery) is determined during phase~1, and the ideals of agents for excess supply objects are not taken into account when determining the welfare of agents in URC mechanisms -- that is, 
\begin{quote}
 for every two profiles $c$ and $c'$ with $ED(c)=ED(c')$, which coincide on $ED(c)$, we have for every agent $i\in N$, for every sequences $\alpha$ and $\beta$,  $\mathbf{d}(URC^{\alpha,\beta}_i(c),c_i)=\mathbf{d} (URC^{\alpha,\beta}_i(c'),c'_i)$.    
\end{quote}
This is because all excess demand objects are exhausted in phase 1, and divided based on the uniform rule. According to Lemma~\ref{lem:distance-sameside} (Equation~(\ref{distance-sameside}): $ \mathbf{d}(p_i,c_i)=2\times \sum_{a\in ED(c)}(c_{ia}-p_{ia})$), as for $c$ and $c'$, we have $ED(c)=ED(c')$, and they coincide on $ED(c)$,  we can conclude  for every agent $i\in N$,  $\mathbf{d}(URC^{\alpha,\beta}_i(c),c_i)=\mathbf{d} (URC^{\alpha,\beta}_i(c'),c'_i)$.

However, there is no   efficient   mechanism that behaves similarly for excess supply objects, and when determining the welfare of agents, the ideals of agents for excess demand objects are not taken into account.

 \begin{proposition}\label{prop:ESimpossible}
     There exists no efficient mechanism $\mu$ such that   for every two profiles $c$ and $c'$ with $ES(c)=ES(c')$, that coincide on $ES(c)$, we have for every agent $i\in N$, $\mathbf{d}(\mu_i(c),c_i)=\mathbf{d} (\mu_i(c'),c'_i)$.
 \end{proposition}
 \begin{proof}
     See Appendix~\ref{app:Mechan}.
 \end{proof}
 
\noindent In URC mechanisms, the share of each agent $i$ to transfer some amount from excess demand objects to excess supply ones is equal to $\sum_{a\in ED(c)} (c_{ia}-URC_{ia}(c_a))$.

\noindent A result of Proposition~\ref{prop:ESimpossible} is that there exists no efficient mechanism where the transfer for each agent $i$ is equal to~$\sum_{a\in ES(c)} (URC_{ia}(c_a)-c_{ia})$.

\vspace{0.1cm}

\noindent 2) Another approach to determining the appropriate amount to transfer is to simplify the problem to a one-object scenario. Consider the following profile for three agents and three objects:
      \begin{itemize}
      \centering
        \item[] $c_1=(a:0.6,b:0.2,c:0.2)$,
        \item[] $c_2=(a:0.5,b:0.4,c:0.1)$,
        \item[] $c_3=(a:0.2,b:0,c:0.8)$.
        
    \end{itemize}
 
Objects $a$ and $c$ are in excess demand. The mechanism designer treat all excess demand objects as one proxy object $d$, with $c_{1d}=c_{1a}+c_{1c}=0.8$, $c_{2d}=c_{2a}+c_{2c}=0.6$, and $c_{3d}=c_{3a}+c_{3c}=1$, when the amount for object $d$ is assumed to be equal to $2$.  Object $d$ is in excess demand as $c_{1d}+c_{2d}+c_{3d}=2.4>2$. 
The mechanism designer employs the uniform rule (UR) for object $d$ and derives 

\begin{center}
    $p_{1d}=0.7, p_{2d}=0.6, p_{3d}=0.7$. 
\end{center}
So agent $1$ must transfer $c_{1d}-p_{id}=0.8-0.7$ from objects $a$ and $c$ to object $b$,
 agent $2$ must transfer $(0.6-0.6)$
 from  objects $a$ and $c$ to object $b$, and agent $3$  must transfer $(1-0.7)$ from objects $a$ and $c$ to object $b$. 
 
 However, this approach is not strategy proof as agent $3$ has incentive to misreport $c'_3=(a:0.2,b:0.1,c:0.7)$. By this misreporting, only object $a$ is in excess demand, and thus $c_{1d}=c_{1a}=0.6$, $c_{2d}=c_{2a}=0.5$, and $c_{3d}=c'_{3a}=0.2$.  The mechanism designer employs the uniform rule (UR) for object $d$ and derives $p_{1d}=0.4, p_{2d}=0.4, p_{3d}=0.2$. In this way, agent $3$ transfers $0$ from object $a$ and $0.1$ from object $c$ 
  to $b$ which is less than $0.3$.

 \section{Characterizing URC }\label{sec:charact}
 In Section~\ref{sec:Mechan}, we introduced  URC mechanisms  and examined their properties. In this section, we characterize URC mechanisms up to welfare equivalence; proving our main Theorem~\ref{charac}, that asserts
  \begin{quote}
     if a mechanism   satisfies   strategy proofness, efficiency, non-bossiness, replacement monotonicity, in-betweenness, and anonymity  then it is welfare equivalence to URC mechanisms.
 \end{quote}
 The main idea of the proof of Theorem~\ref{charac} is to demonstrate that every mechanism satisfying the properties outlined in Theorem~\ref{charac} divides the chances of all excess demand objects according to the uniform rule. Our approach is as follows: given a profile $c$, for every $a\in ED(c)$, we transform $c$ into another profile that includes only object $a$ as its excess demand object. Subsequently, we will employ the characterization theorem established by Y. Sprumont in 1999 \cite{SPRUMONT1991Econometrica}, page 511.
 
\noindent The proof sketch is outlined as follows:
Suppose that an arbitrary mechanism $\mu$ satisfying the properties outlined in Theorem~\ref{charac} is given.
 \begin{itemize}
     \item[Step.1.] First, in Lemma~\ref{prop-CSP}, using efficiency, strategy proofness and replacement monotonicity, we prove for every profile $c$ with only one excess demand object $a$, where $ED(c)=\{a\}$,  if an agent $j$ misreports solely on non-excess demand objects without affecting the set $ED(c)$, then all agents' allocation for the excess demand object $a$ remains unchanged.
\end{itemize}

 \begin{lemma}\label{prop-CSP}
   Let $\mu$ be an   efficient, strategy proof and replacement monotonic   mechanism. Suppose  $c,c'\in C^n$ are two profiles where for some object $a\in A$,  and some agent $j\in N$,   
   \begin{itemize}
   \item $c'=(c'_j,c_{-j})$, that is, for all $l\in N\setminus \{j\}$, $c'_l=c_l$, 
   \item  $ED(c)=ED(c')=\{a\}$,       and
        \item  $c_{ja}=c'_{ja}$.
   \end{itemize}
  Then     for all $i\in N$,  $\mu_{ia}(c)=\mu_{ia}(c')$.
  \end{lemma}
\begin{proof}
    See Appendix~\ref{app:charact}.
\end{proof}
 
\begin{itemize}
     \item[Step.2.]   Subsequently, we prove Lemma~\ref{lem:Twoprofiles}. This lemma asserts that when considering two preference profiles, where a single object, denoted as $a$, is their only excess demand object, and these profiles are identical in their ideal peaks of object $a$, then the outcomes of any strategy proof, efficient, and replacement monotonic mechanism for these two profiles will also be identical for object $a$.

\end{itemize}
  \begin{lemma}\label{lem:Twoprofiles}

 Let $\mu$ be a strategy proof, efficient, and replacement monotonic. For every two profiles $c$ and $c'$ where for some $a\in A$
\begin{itemize} \item $ED(c)=ED(c')=\{a\}$ , and 
\item $c_a=c'_a$ ($c_{ia}=c'_{ia}$ for all $i\in N$), 
\end{itemize}
we have 
\begin{center} for all $i\in N$,
$\mu_{ia}(c)=\mu_{ia}(c')$.  
\end{center}
     
 \end{lemma}
\begin{proof}
    See Appendix~\ref{app:charact}.
\end{proof}

\begin{itemize}

\item[Step.3.]  By employing the characterization theorem established by Y. Sprumont, 1999 \cite{SPRUMONT1991Econometrica}, page 511, we deduce that for every  profile $c$ with only one excess demand object, say $a$,  for every agent $i\in N$, $\mu_{ia}(c)=UR_i(c_a)$.
\end{itemize}
\begin{lemma} \label{lemEDa}
  Let  $\mu$ be a strategy proof, efficient, replacement monotonic, and anonymous mechanism. Also, let   $c$  be a profile with $ED(c)=\{a\}$ for some $a\in A$. Then,  for every  $i\in N$,
  
  \begin{center}
       $\mu_{ia}(c)=UR_i(c_a)$,
  \end{center} and 
             $\mathbf{d}(c_i,\mu_i(c))=2\times (c_{ia}-UR_i(c_a))= \mathbf{d}(c_i,URC_i(c))$.

\end{lemma}
\begin{proof}
    See Appendix~\ref{app:charact}.
\end{proof}

\begin{itemize}
   
 \item[Step.4.] Ultimately, with the aid of  Corollary~\ref{cor:between}, given a profile $c$, for every $a\in ED(c)$, we transform  the profile $c$ into another profile $c'$ featuring only a single excess demand object $ED(c')=\{a\}$, and  for every $i\in N$, $\mu_{ia}(c)=\mu_{ia}(c')$ . This transformations enable us to apply Lemma~\ref{lemEDa}, and conclude the main Theorem~\ref{charac} which characterizes URC mechanisms up to welfare equivalency.
 \end{itemize}

 \begin{theorem}\label{charac}
If a mechanism $\mu$ is strategy proof, efficient, replacement monotonic,   non-bossy, in-between and anonymous then for every profile $c\in C^n$, for every agent $i\in N$, for all $a\in ED(c)$, $\mu_{ia}(c)=UR_i(c_a)$, and 
\begin{center}
   for all $i\in N$, $\mu_i(c)\equiv_{c_i} URC_i(c)$,
\end{center}  
  \end{theorem}
  \begin{proof}
    See Appendix~\ref{app:charact}.
\end{proof}

\noindent The reader can refer to Figure~\ref{fig:URC} for a graphical illustration of Theorem~\ref{charac}.

 \begin{remark}\label{remarkcar}
     Theorem~\ref{charac} remains valid if we replace the anonymity property with envy freeness. The proof follows a similar argument, with the only difference being that in the proofs, we rely on the characterization theorem established by Y. SPrumont, 1999, \cite{SPRUMONT1991Econometrica}, page 517, for envy freeness, strategy proofness, and efficiency.
 \end{remark}

\section{Logical Relationship between Properties}\label{sec:indaxiom}
We discussed properties strategy proofness, efficiency, replacement monotonicity, non-bossiness, anonymity (and envy freeness), and show that every mechanism that satisfies these properties is welfare equivalent to URC mechanisms (Theorem~\ref{charac}, and Remark~\ref{remarkcar}). In this section, we study logical Independency of these properties. To investigate the logical relationship between properties, We introduce some other mechanisms for the division problem of chances: a class of Serial Dictatorship mechanisms for dividing chances (SDC mechanisms), and a class of Proportional Division of Chances mechanisms (PDC mechanisms).

 \subsection{SDC Mechanisms}

Serial Dictatorship mechanisms for dividing chances (SDC mechanisms) consist of two phases. In Phase 1, all agents are arranged in a line, and then each agent, in their turn, takes out the amount closest to their preference from each object. Phase 2 of SDC mechanisms is exactly the same as Phase 2 of URC mechanisms. Let $\alpha$ represent a sequence of all agents in $N$, and $\beta$ a sequence of all objects in $A$. Given a profile $c=(c_i)_{i\in N}$, the mechanism $SDC^{\alpha,\beta}$   operates in  two phases:

\begin{itemize}
\item \textbf{Phase 1}:
\begin{itemize}
    \item[0.]    Initialize $t=1$ and let the $r_a$ represents the remainder   of each object $a\in A$, and initially $r_a=1$. 
    \item[1.]  For all $a\in A$, let $w_{\alpha(t)a}=\min(r_a, c_{\alpha(t)a})$ (agent $\alpha(t)$, in their turn, either extracts their ideal $c_{\alpha(t)a}$ from the remainder of object $a$, or, if the remainder is less than their ideal, then they extract the whole remainder). 
    \item[2.] For every object $a\in A$, update its remainder $r_a\leftarrow (r_a-c_{\alpha(t)a})$.

    \item[3.] if $t=n$, Stop, otherwise move to next agent by updating $t$ to $t+1$ and return to Step 1.
\end{itemize}
 \item \textbf{Phase 2}: Phase 2 of SDC mechanisms is executing Phase 2 of URC mechanisms for sequences $\alpha$ and $\beta$.
      
\end{itemize}
Finally, we let for every $i\in N$, for every $a\in A$, $SDC^{\alpha,\beta}_{ia}(c)=w_{ia}$ as the outcome of the mechanism.

Note that the distances between the allocations and ideal lotteries are contingent upon the excess demand objects. Since all excess demand objects are exhausted in phase 1 of SDC mechanisms, the determination of distances takes place during this phase. Subsequently, phase 2 serves primarily to ensure the feasibility properties.

\begin{example}\label{exp-SD}
For $N=\{1,2,3\}$ and $A=\{a,b,c\}$, suppose the preference profile (\ref{eq:exam1}) is given.  Consider  the sequence $\alpha$ of agents  is as follows: agent 1 is ahead, agent 2 is  next, and agent 3 is in the third position, and $\beta=abc$. 

\noindent We run the phase 1 of $SDC^{\alpha,\beta}$:
Agent 1 receives   $w_{1a}=0.6,  w_{1b}=0.2$ and $w_{1c}=0.2$. Hence the remaining    of objects are:
\begin{center}
    $r_a=0.4$, $r_b=0.8$, and $r_c=0.8$.
\end{center}
Then it is the turn of agent 2. He takes   $w_{2a}=0.4$, $w_{2b}=0.4$, and $w_{2c}=0.1$. 
The remaining   of objects are:
\begin{center}
    $r_a=0$, $r_b=0.4$, and $r_c=0.7$.
\end{center}
Next, agent 3 takes $w_{3a}=0$, $w_{3b}=0$, and $w_{3c}=0.7$.
Thus the remaining  of objects are:
\begin{center}
    $r_a=0$, $r_b=0.4$, and $r_c=0$.
\end{center}
Finally, in phase 2, for sequencing $\beta$ of objects,  the remaining   are given to those agents who are not yet full. Hence the final random matching is:
\begin{itemize}
\item  $(p_{1a}=0.6, p_{1b}=0.2,p_{1c}=0.2)$,
    \item $(p_{2a}=0.4,p_{2b}=0.4+0.1,p_{2c}=0.1)$, and
    \item $(p_{3a}=0,p_{3b}=0+0.3,p_{3c}=0.7)$.
    
\end{itemize}
 
\end{example}

\subsection{PDC Mechanisms}

 Proportional Division of Chances mechanisms (PDC mechanisms) consist of two phases. In phase 1, each agent is given their ideal for non-excess demand objects, and excess demand objects are divided proportional to the ideals of agents. 
Phase 2 of PDC mechanisms is exactly the same as Phase 2 of URC mechanisms. 
Let $\alpha$ represent a sequence of all agents in $N$, and $\beta$ a sequence of all objects in $A$. Given a profile $c=(c_i)_{i\in N}$,  the mechanism $PDC^{\alpha,\beta}$ operates in two phases as follows:

\begin{itemize}
    \item \textbf{Phase 1}:  For every object $a\in A$, 
    \begin{itemize}

     \item if $a\not\in ED(c)$, for all $i\in N$, let $w_{ia}=c_{ia}$,
        \item if $a\in ED(c)$, for all $i\in N$, let $w_{ia}=c_{ia}/(\sum_{j\in N}c_{ja})$.
        
    \end{itemize}
    
    \item \textbf{Phase 2}:
Phase 2 of PDC mechanisms is identical to Phase 2 of URC mechanisms.
\end{itemize}
Finally, we let for every $i\in N$, for every $a\in A$, $PDC^{\alpha,\beta}_{ia}(c)=w_{ia}$ as the outcome of the mechanism.
Similar to URC mechanisms and SDC mechanisms, phase 2 of PDC mechanisms serves primarily to ensure the feasibility properties.

\begin{proposition}\label{prop:PDSD}~

\begin{itemize}
    \item[i)] SDC mechanisms are efficient, strategy-proof, replacement monotonic, in-between and non-bossy, but they are neither anonymous nor envy-free.
    \item[ii)] PDC mechanism are anonymous, efficient, replacement monotonic, in-between, and non-bossy but they are neither strategy proof nor envy free.
     
\end{itemize}
 \end{proposition}
    \begin{proof}
    See Appendix~\ref{app:indaxion}.
\end{proof}

 \noindent In addition to Proposition~\ref{prop:PDSD}, following logical relations between properties also hold  true.  Let $N=\{1,2,3\}$ and $A=\{a,b,c\}$.
 
\begin{itemize}
    
    \item[iii)] \textit{Non-bossiness + strategy proofness + efficiency does not imply replacement monotonicity.}
    
    To show this claim, we introduce a mechanism denoted as $Except$. Let $\alpha=23$ , and $\beta=abc$. Given a profile $c=(c_i)_{i\in N}$, the  mechanism $Except^{\alpha,\beta}$ assigns agent 1 his ideal lottery. Then if the ideal lottery of agent 1 is $(c_{1a}=1/3)_{a\in A}$, the mechanism proceeds with serial dictatorship for sequences $\alpha$, and $\beta$, i.e., $SDC^{\alpha,\beta}$.  Otherwise it proceeds with serial dictatorship for sequences $\alpha^r$, and $\beta$, i.e., $SDC^{\alpha^r,\beta}$ where $\alpha^r=32$ is the reverse of $\alpha$.

    The $Except^{\alpha,\beta}$ is not replacement monotonic. Let 
    \begin{itemize}
    \centering
        \item[] $c_1=(a:0.4,b:0.4,c:0.2)$
        \item[] $c_2=(a:0.4,b:0.4,c:0.2)$
        \item[] $c_3=(a:0.4,b:0.4,c:0.2)$.
        
    \end{itemize} Let $c'_1= (a:1/3,b:1/3,c:1/3)$. We have
\begin{itemize}
    \item $ED(c)=ED((c'_1,c_{-1}))=\{a,b\}$, and
    \item for all object $o\in ED(c)=\{a,b\}$, $c'_{1o}\leq c_{1o}$.
\end{itemize}
The outcome of the Except mechanism for agent 1 is as follows:
\begin{itemize}
    \item[] $Except^{\alpha,\beta}_1(c))=(a:0.4,b:0.4,c:0.2)$, and 
    \item[] $Except^{\alpha,\beta}_1((c'_1,c_{-1})=(a:1/3,b:1/3,c:1/3)$. 
\end{itemize}
Thus, we have $\mathbf{d}(c_1, Except^{\alpha,\beta}_1(c))=0<\mathbf{d}(c_1, Except^{\alpha,\beta}_1((c'_1,c_{-1}))$. By Definition~\ref{DEF:RM} of replacement monotonicity, we must have for $j\in\{2,3\}$, $\mathbf{d}(c_j, Except^{\alpha,\beta}_j(c))\geq \mathbf{d}(c_j, Except^{\alpha,\beta}_j((c'_1,c_{-1}))$. However, for agent 2, $\mathbf{d}(c_2, Except^{\alpha,\beta}_2(c))=0$ and $\mathbf{d}(c_2, Except^{\alpha,\beta}_2((c'_1,c_{-1}))>0$, as agent $2$ is at the end of the sequence $\alpha^r=32$.

The proofs of strategy proofness, efficiency, and non-bossiness for the Except mechanism are similar to the corresponding proofs for the SDC mechanisms. The proofs of strategy-proofness, efficiency, and non-bossiness for the Except mechanism are similar to the corresponding proofs for the SDC mechanisms. Concerning non-bossiness, note that although the first agent in the sequence can change the order of agents after himself, by misreporting to $c'_1= (a:1/3, b:1/3, c:1/3)$, however, he also changes his own allocation on excess demand objects.

\item[iv)] \textit{Replacement monotonicity does not imply non-bossiness.}
We introduce a mechanism denoted as $ME$ which is replacement monotonic but not non-bossy. Given a profile $c=(c_i)_{i\in N}$, the mechanism $ME$ operates as follows:
\begin{itemize}
    \item If $ED(c)=\{a\}$ then $ME_{1a}(c)=1, ME_{2b}(c)=1, ME_{3c}(c)=1 $.
    \item Otherwise $ME_{1b}(c)=1, ME_{2a}(c)=1, ME_{3c}(c)=1 $.
\end{itemize}
The mechanism $ME$ operates  in such a way  that for two profiles $c_1$ and $c_2$, if we have $ED(c_1)=ED(c_2)$ then $ME(c_1)=ME(c_2)$. Therefore, the concept of replacement monotonicity, as defined in Definition~\ref{DEF:RM}, holds for the $ME$ mechanism. However, the $ME$ mechanism is not non-bossy, as agent 3 can change the outcomes for other agents without changing his own outcome.

\item[v)] \textit{We construct a mechanism, called $MEU$, that is efficient and welfare equivalent to URC mechanisms but is not strategy proof.}

Consider the preference profile $e=(e_1,e_2,e_3)$:
\begin{itemize}
\centering
    \item[] $e_1=(a:2/3,b:1/3, c:0)$ 
    \item[] $e_2=(a:1/3,b:2/3, c:0)$ 
    \item[] $e_3=(a:1/3,b:1/3, c:1/3)$ 
\end{itemize}
Let $\alpha=123$ and $\beta=abc$. We define a mechanism $MEU$ as follows:
 for every profile $c$, if $c\neq e$ then $MEU(c)=URC^{\alpha,\beta}(c)$, and for $c=e$, let $MEU_{1}(e)=(a:2/3,b:0,c:1/3) $,
       $MEU_{2}(e)=(a:0,b:2/3,c:1/3) $,
        $MEU_{3}(e)=(a:1/3,b:1/3,c:1/3) $.
 
It is easy to check that the mechanism $MEU$ is welfare equivalent to $URC^{\alpha,\beta}$, as 
$URC^{\alpha,\beta}_1(e)=URC^{\alpha,\beta}_2(e)=URC^{\alpha,\beta}_3(e)=(a:1/3,b:1/3,c:1/3)$, and 
\begin{itemize}
    \item[] $\mathbf{d}(e_1, URC^{\alpha,\beta}_1(e))=\mathbf{d}(e_1, MEU_1(e))=2/3$, 
    \item[] $\mathbf{d}(e_2, URC^{\alpha,\beta}_2(e))=\mathbf{d}(e_2, MEU_2(e))=2/3$, and
    \item[] $\mathbf{d}(e_3, URC^{\alpha,\beta}_3(e))=\mathbf{d}(e_3, MEU_3(e))=0$. 
\end{itemize}

The mechanism $MEU$ is not strategy proof, as 
 agent 1 in profile $c=(c_1,e_2,e_3)$ where
 $c_1=(a:2/3,b:0, c:1/3)$, has incentive to misreport $c'_1=e_1$.

Note that the existence of $MEU$ does \textbf{not} conflict with Theorem~\ref{charac}. If a mechanism satisfies the properties outlined in Theorem~\ref{charac}, it is welfare equivalent to URC mechanisms, but the reverse is not necessarily true (see Figure~\ref{fig:URC} and Figure~\ref{fig:WEL}).
\end{itemize}
 \begin{figure}[h!]
 \centering
 \includegraphics[scale=0.7]{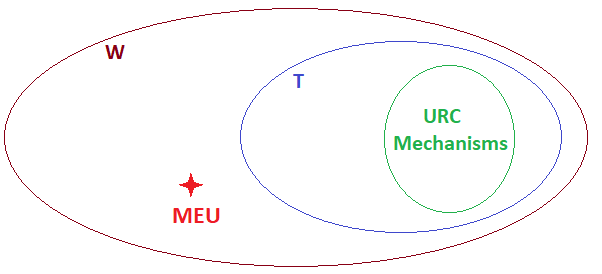} 
  \caption{URC Mechanisms}
 { \footnotesize
 W: the set of all mechanisms that are welfare equivalent to URC mechanisms.
 
  T: the set of all mechanisms that satisfy the properties outlined in Theorem~\ref{charac}. 
  
  Proposition~\ref{Theorem:URL-properites} says: $URC \subseteq T$, and Theorem~\ref{charac} says: $T\subseteq W$.
  }
  \label{fig:URC}
\end{figure}

   \begin{figure}[h!]
 \centering
 \includegraphics[scale=1.2]{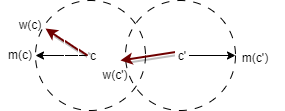} 
  \caption{Welfare Equivalency}
 { \footnotesize
Assume $w$ and $m$ as two mechanisms. These two mechanisms are welfare equivalent since $d(c,m(c))=d(c,w(c))$ and $d(c',m(c'))=d(c',w(c'))$.
The mechanism $m$ is strategy proof, since $d(c,m(c))<d(c,m(c'))$ and $d(c',m(c'))<d(c',m(c))$. However, the mechanism $w$ is not strategy proof as $d(c,w(c))>d(c,w(c'))$.
  }
  \label{fig:WEL}
\end{figure}

    

In Table~\ref{table:1} we illustrate how introduced mechanisms fulfill properties. The  Equal-Division mechanism in table~\ref{table:1}, is simply a mechanism that regardless of the given preference profile, for all agents $i\in N$ and objects $a\in A$, assigns each agent $i\in N$ a share of $p_{ia}=1/n$ of object $a$.

The following propositions are yet unknown and pose  open questions for us.
\begin{itemize}
\item Is there a mechanism satisfying strategy proofness and efficiency  but not non-bossiness?
    \item Is there a mechanism satisfying strategy proofness and efficiency  but not in-betweenness?
    
\item  Although, we explored  alternative mechanisms for dividing chances, in Sections~\ref{subsec:explor} and \ref{sec:indaxiom}, yet the following question~remains unsolved:
\begin{quote}
\textit{Is there a mechanism that is strategy proof, efficient, and satisfies either anonymity or envy freeness, yet is not welfare equivalent to URC mechanisms? }   
\end{quote} The likelihood of affirming `Yes' diminishes based on Proposition~\ref{prop:ESimpossible}. The SDC mechanisms are strategy proof and efficient, and not welfare equivalent to URC mechanisms, but they are not fair and lack both anonymity and envy freeness. Theorem~\ref{charac}, which characterizes URC mechanisms in terms of welfare equivalence, doesn't provide a conclusive answer to our open question because it considers additional properties such as non-bossiness, in-betweenness, and replacement monotonicity.

\end{itemize}

   \begin{table}[]
   \centering
\begin{tabular}{|l|l|l|l|l|l|l|l|l}
\hline
 Mechanisms/Properties& SP & PF & RM  & NB & IB & ANO & EF  \\ \hline
URC & \checkmark & \checkmark & \checkmark&\checkmark & \checkmark  &\checkmark & \checkmark \\ \hline

SDC & \checkmark &\checkmark &\checkmark & \checkmark &  \checkmark  & $\times$ & $\times$ \\ \hline

PDC & $\times$ &\checkmark & \checkmark & \checkmark & \checkmark   & \checkmark & $\times$ \\ \hline

Equal-Division & \checkmark  & $\times$ & \checkmark &\checkmark &  \checkmark    & \checkmark & \checkmark \\ \hline





\end{tabular}
\caption{ Mechanisms/Properties}
\label{table:1}
\footnotesize{SP: Strategy Proofness, PF: (Pareto) Efficiency, RM: Replacement Monotonicity NB: Non-Bossiness,

IB: In-Betweeness, ANO:~Anonymity, EF: Envy Freeness, \checkmark: Yes, $\times$: No.   }
\end{table}

\section{Concluding Remarks and Further Works}\label{sec:conclu}
 We delved into frequently repeated matching scenarios where individuals seek diversification in their choices, and their favored option is not a specific outcome but rather a lottery over them, representing the peak of their preferences. Subsequently, we introduced a class of mechanisms known as URC mechanisms designed for dividing chances in repeated matching problems. We then established a characterization theorem up to welfare equivalence, demonstrating that any mechanism satisfying Pareto efficiency, strategy proofness, replacement monotonicity, non-bossiness, in-betweenness, and anonymity (or envy freeness) is welfare equivalent to URC mechanisms. In our exploration of alternative approaches in Sections~\ref{subsec:explor} and \ref{sec:indaxiom}, a fundamental question remains open: \textit{Can a mechanism be both strategy-proof and efficient while adhering to either anonymity or envy-freeness, and still not be welfare-equivalent to URC mechanisms?}

In this paper, we addressed the problem of dividing chances using ideal lotteries to represent preferences for one-sided, one-to-one matching. As a potential avenue for future research, we could explore two-sided markets, such as the \textit{marriage problem} and the \textit{roommate problem}, where agents' preferences are also represented using ideal lotteries.

 We can also consider extending the concept of ideal lotteries to ideal Markov chains. In the introduction, we discussed an example involving a collection of music on a smartphone, where an individual's favorite option is not a specific music record, but rather a lottery over them. Taking this a step further, we can argue that their preferred option is not merely a lottery but a Markov chain. This Markov chain can be learned by the application's artificial intelligence based on collected data, including how songs are replayed and transitions between songs within their collection. For instance, consider a collection of four songs, $\{a, b, c, d\}$. A favorite option could be represented as a Markov chain, where each song is a state, and transitions between songs occur with certain probabilities. For example, an ideal Markov chain for an agent might show that after listening to song `a', the agent would like to replay `a' with a probability of 0.5, transition to `b' with a probability of 0.3, and switch to `c' with a probability of 0.2.

\begin{center}
  \includegraphics[scale=0.4]{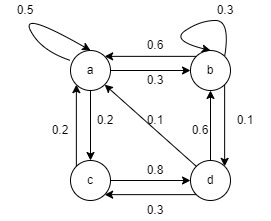} 
  
 ~~ An Ideal Markov Chain of Songs\label{figure} 
\end{center}
Also, consider the example of a company with two workers and two tasks that hourly repeated, where, each worker has an ``ideal Markov chain" over tasks that represent their favorite option.  Let $N=\{1,2\}$ and $A=\{a,b\}$. The ideal Markov chain for agent 1 is shown by $M1$ and for agent 2 by $M2$.
\begin{center}
 \includegraphics[scale=0.45]{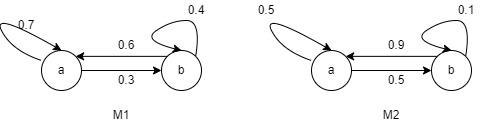} 
  
 ~~Ideal Markov Chains of workers\label{figure2} 
\end{center}
After doing task $a$, agent 1, with probability $0.7$, would like to do task $a$ again in the next hour, and with probability $0.3$, would like to do task $b$ in the next hour. Also, after doing task $b$, agent 1, with probability $0.4$, would like to do task $b$ in the next hour, and with probability $0.6$ do task $a$.

 Representing agents' preferences through ideal Markov chains finds application in designing recommender systems (Aggarwal, 2016 \cite{RecomBook}), particularly to address users' 'desire for variety.' Mechanism design becomes especially intriguing when agents' preferences are modeled using advanced techniques such as \textit{Markov chains}, \textit{recurrent neural networks (RNN)}, and \textit{long short-term memory networks (LSTM)}. These models capture individual preferences over a set of objects in a dynamic and evolving manner, departing from simple linear orderings.

To illustrate this concept, consider a scenario with four friends on a road trip, sharing a car and a music collection for their journey. Each person in the car has their own Markov chain representing song preferences and how they want songs repeated during the trip. The challenge is to develop an algorithm for the car's music player that aggregates individual Markov chains, creating a coherent \textbf{social Markov chain} to maximize overall passenger utility.

This challenge extends to platforms like \textbf{Spotify}, where dynamic user preferences learned by RNN or LSTM models need effective aggregation mechanisms. Our future work explores developing mechanisms and algorithms with potential benefits for platforms like \textbf{Netflix} and \textbf{Spotify}. We also delve into questions of fairness, incentive compatibility, stability, and other considerations in recommender systems regarding Markov chain modeling of preferences. For instance, our further research  may enhance job satisfaction by optimizing matching algorithms on \textbf{freelance platforms}.

{\small \ }
\printbibliography

\section{Appendix}\label{sec:appen} 

In the appendix, we provide proofs of lemmas, propositions, and theorems that have not been addressed in earlier sections. Additionally, we offer some examples that are referenced in preceding sections.

\subsection{Examples} \label{app:Example}

\begin{example}\label{exm:divisionchanc}
 This example illustrates how the division of chances in matching scenarios \textbf{differs} from the division of quantities of multiple commodities due to the concept of \textit{agent feasibility}. Applying the generalized  uniform rule is not applicable as a result.
 
    Suppose that $N=\{1,2,3\}$ is a set of agents, $A=\{a,b,c\}$ is a set of objects, and ideal lotteries of agents over objects are
  
      \begin{itemize}
      \centering
        \item[] $c_1=(a:0.2,b:0.6,c:0.2)$,
        \item[] $c_2=(a:0.4,b:0.6,c:0)$,
        \item[] $c_3=(a:0,b:0.2,c:0.8)$.
        
    \end{itemize}
 
    Object $a$ is in excess supply, $(\sum_{i\in N}c_{ia}<1)$. If we divide the chance of receiving  object $a$ using the uniform rule then we have $p_{1a}=0.3$, $p_{2a}=0.4$, and $p_{3a}=0.3$. Dividing the chance of  object $c$ using the uniform rule, we have $p_{1c}=c_{1c}=0.2$, $p_{2c}=c_{2c}=0$, and $p_{3c}=c_{3c}=0.8$. 

    For agent $3$, we have $p_{3a}+p_{3c}=0.3+0.8>1$, which contradicts  agent feasibility, that is, $(p_{3a}+p_{3b}+p_{3c}=1)$.
    
\end{example}

\begin{example}\label{Exam:non} This example illustrates that URC mechanisms are not welfare non-bossy.
    Suppose that $N=\{1,2,3\}$ is a set of agents, $A=\{a,b,c\}$ is a set of objects, and ideal lotteries of agents over objects are
  
      \begin{itemize}
      \centering
        \item[] $c_1=(a:0.3,b:0.5,c:0.2)$,
        \item[] $c_2=(a:0.7,b:0.2,c:0.1)$,
        \item[] $c_3=(a:0.1,b:0.4,c:0.5)$.
        
    \end{itemize}
    Let $\alpha=123$ and $\beta=abc$. If we run  $URC^{\alpha,\beta}$ on the profile $c$, we have the following outcome
     \begin{itemize}
    
        \item[] $p_1=(a:0.3,b:0.4,c:0.3)$,
        \item[] $p_2=(a:0.6,b:0.2,c:0.2)$,
        \item[] $p_3=(a:0.1,b:0.4,c:0.5)$.
        
    \end{itemize}
    Let $c'_2=(a:0.7,b:0.3,c:0)$. The outcome of $URC^{\alpha,\beta}$ on the profile $(c_1,c'_2,c_3)$ is 
 \begin{itemize}
    
        \item[] $p'_1=(a:0.3,b:0.35,c:0.35)$,
        \item[] $p'_2=(a:0.6,b:0.3,c:0.1)$,
        \item[] $p'_3=(a:0.1,b:0.35,c:0.55)$.
        
    \end{itemize} We have $\mathbf{d}(c_2,p_2)=\mathbf{d}(c_2,p'_2)=0.2$. However, $\mathbf{d}(c_3,p_3)\neq \mathbf{d}(c_3,p'_3)$.
    
\end{example}

 \subsection{Proofs for Section~\ref{sec:Model}} \label{app:Model}

\begin{proof}   Proof of Proposition~\ref{eff-same}:
   \hspace{0.2cm}

   \textbf{Proof of if}:
    If $P$ is an efficient random matching then it is not strictly lottery dominated by any other random matching and thus by Lemma~\ref{lemregualr}, it is same-sided.

\textbf{Proof of only if:}
 Suppose that $P$ is same-sided but not efficient. Then $P$ is strictly lottery dominated by another random matching $J$. Either $J$ is same-sided or not; if not, then by Lemma~\ref{lemregualr}, there exists a same-sided random matching $Q$ that strictly lottery dominates $J$ and thus strictly lottery dominates $P$. So, there exists a same-sided random matching $Q$   that strictly lottery dominates $P$. 
 
 As both $P$ and $Q$ are same-sided, for every object $a\in A$,
 \begin{itemize}
     \item if it is $ED(c)$ then    for all $i\in N$ $p_{ia}\leq c_{ia}$ and  $q_{ia}\leq c_{ia}$,
     \item  if it is $ES(c)$ then  for all $i\in N$ $p_{ia}\geq c_{ia}$ and  $q_{ia}\geq c_{ia}$, and 
     \item  if it it is $UN(c)$ then  for all $i\in N$ $p_{ia}= c_{ia}$ and  $q_{ia}= c_{ia}$.
 \end{itemize}
 Since $Q$ strictly lottery dominates $P$ there exists an agent $i$ such that $\mathbf{d}(q_i,c_i)<\mathbf{d}(p_i,c_i)$.

 We have
 \begin{itemize}
     \item $d(p_i,c_i)=\sum_{a\in ED(c)} (c_{ia}-p_{ia}) +\sum_{a\in ES(c)} (p_{ia}-c_{ia})$
      \item $d(q_i,c_i)=\sum_{a\in ED(c)} (c_{ia}-q_{ia}) +\sum_{a\in ES(c)} (q_{ia}-c_{ia})$
 \end{itemize}
 so, $\sum_{a\in ED(c)} (c_{ia}-p_{ia}) +\sum_{a\in ES(c)} (p_{ia}-c_{ia})> \sum_{a\in ED(c)} (c_{ia}-q_{ia}) +\sum_{a\in ES(c)} (q_{ia}-c_{ia})$ and  thus
 \begin{center}
     $\sum_{a\in ED(c)}(q_{ia}-p_{ia})+ \sum_{a\in ES(c)}(p_{ia}-q_{ia})>0$.
 \end{center}
 Also, for all other agents $j\in N/\{i\}$, we have  $\mathbf{d}(q_j,c_j)\leq\mathbf{d}(p_j,c_j)$ and thus 
 \begin{center}
     $\sum_{a\in ED(c)}(q_{ja}-p_{ja})+ \sum_{a\in ES(c)}(p_{ja}-q_{ja})\geq0$
 \end{center}
 Therefore we have
 \begin{center}
     $\sum_{j\in N} (\sum_{a\in ED(c)}(q_{ja}-p_{ja})+ \sum_{a\in ES(c)}(p_{ja}-q_{ja}))>0$
 \end{center}
 which implies
 \begin{center}
     $ \sum_{a\in ED(c)}(\sum_{j\in N}q_{ja}-\sum_{j\in N}p_{ja})+ \sum_{a\in ES(c)}(\sum_{j\in N}p_{ja}-\sum_{j\in N}q_{ja})>0$
 \end{center}
 \begin{center}
     $ \sum_{a\in ED(c)}(1-1)+ \sum_{a\in ES(c)}(1-1)>0$
 \end{center}
 Contradiction.
\end{proof}

\begin{lemma}\label{lemregualr}
    For every random matching $P$, either $P$ is same-sided or there exists a same-sided random matching $Q$ that strictly lottery dominates $P$.
\end{lemma} 
  \begin{proof}\label{Prooflemregualr} Proof of Lemma~\ref{lemregualr}:
    Let $P$ be an arbitrary random matching. If $P$ is same-sided we are done. Else, there exists a tuple $(b_0,i,j)$ such that $b_0\in A$, $i,j\in N$ and $p_{ib_0}<c_{ib_0}$ and $p_{jb_0}>c_{jb_0}$. 
    As $p_{ib_0}<c_{ib_0}$, and $\sum_{a\in A}p_{ia}=\sum_{a\in A}c_{ia}=1$  there exists an object $b_1$ such that $p_{ib_1}>c_{ib_1}$. If $p_{jb_1}<c_{jb_1}$ let $\alpha=c_{jb_1}$ else let $\alpha=1$.
    Consider 
    \begin{center}
        $0<\epsilon=\min(c_{ib_0}-p_{ib_0},p_{jb_0}-c_{jb_0},p_{ib_1}-c_{ib_1}, \alpha-p_{jb_1} )$.
    \end{center}
     For agent $i$, 
   let $q_{ib_0}=p_{ib_0}+\epsilon$, $q_{ib_1}=p_{ib_1}-\epsilon$, and for all $a\in A/\{b_0,b_1\}$, $q_{ia}=p_{ia}$.
   For agent $j$,  let $q_{jb_0}=p_{jb_0}-\epsilon$, $q_{jb_1}=p_{jb_1}+\epsilon$, and for all $a\in A/\{b_0,b_1\}$, $q_{ja}=p_{ja}$.
   For any agent $i'\in N/\{i,j\}$ let $q_{i'a}=p_{ia}$ for all $a\in A$.

   For agent $j$, 
  \begin{center}
      $|q_{jb_0}-c_{jb_0}|+|q_{jb_1}-c_{jb_1}|=
      |p_{jb_0}-\epsilon-c_{jb_0}|+|p_{jb_1}+\epsilon-c_{jb_1}|=$
      $ |p_{jb_0}-c_{jb_0}|-\epsilon + |p_{jb_1}-c_{jb_1}+\epsilon| $ 
      $ \leq |p_{jb_0}-c_{jb_0}|+ |p_{jb_1}-c_{jb_1}|$.
  \end{center} Therefore, $\mathbf{d}(c_j,q_j)\leq\mathbf{d}(c_j,p_j)$. Also, for any agent $i'\in N/\{i,j\}$, $\mathbf{d}(c_{i'},q_{i'})=\mathbf{d}(c_{i'},p_{i'})$. For agent $i$, we have $\mathbf{d}(c_i,q_i)<\mathbf{d}(c_i,p_i)$, and thus $Q$ strictly lottery-dominates $P$.

  The value of $\epsilon$ is equal to one the values $c_{ib_0}-p_{ib_0},p_{jb_0}-c_{jb_0},p_{ib_1}-c_{ib_1}$, or $\alpha-p_{jb_1}$. So for the random matching $Q$, at least one of the following equalities hold: 
  
  \begin{equation}\label{eq:reg1}
      q_{ib_0}=c_{ib_0},~q_{ib_1}=c_{ib_1},~q_{jb_0}=c_{jb_0}, ~q_{jb_1}=c_{jb_1}~ or~q_{jb_1}=1.
  \end{equation}

If $Q$ is same-sided, then we are done. Otherwise, we repeat the above process for $Q$ until we obtain a same-sided bistochastic matrix. At each repetition, at least one of the equalities in (\ref{eq:reg1}) holds true for the obtained $Q$ and some agents $i,j\in A$, which guarantees that we cannot repeat the process an infinite number of times. After a finite number of repetitions, we will finally reach a bistochastic matrix $Q$ that is same-sided and strictly lottery dominates $P$.     
\end{proof}

\begin{proof} Proof of Proposition~\ref{prop:Between}.

\noindent Since $c'_i$ is between $c_i$ and $\mu_i(c)$, and $\mu$ is same-sided, we have
 for all $a\in A$
    \begin{itemize}
        \item if $a\in ED(c)$ then $\mu_{ia}(c)\leq c'_{ia}\leq c_{ia}$,
        \item if $a\in UN(c)$ then $c'_{ia}= c_{ia}$, and
        \item if $a\in ES(c)$ then $c_{ia}\leq c'_{ia}\leq \mu_{ia}(c)$.
    \end{itemize}
\noindent Since $\mu$ is same-sided (efficient), we have 
  
\noindent    $\mathbf{d}(c_i,\mu_i(c))=  \sum_{a\in ED(c)}(c_{ia}-\mu_{ia}(c))+ \sum_{a\in ES(c)}(\mu_{ia}(c)-c_{ia})$

\noindent    $= \sum_{a\in ED(c)}\bigg( (c_{ia}-c'_{ia})+(c'_{ia}-\mu_{ia}(c))\bigg)+      
      \sum_{a\in ES(c)}\bigg((\mu_{ia}(c)-c'_{ia})+(c'_{ia}-c_{ia})\bigg)  $ 

 \noindent     $= \mathbf{d}(c_i,c'_i)+  \mathbf{d}(c'_i,\mu_i(c))$.

 \noindent     Therefore, 
  \begin{equation}\label{eq:c1}
      \mathbf{d}(c_i,\mu_i(c))=  \mathbf{d}(c_i,c'_i)+  \mathbf{d}(c'_i,\mu_i(c)).     
   \end{equation} 
 \noindent  By strategy proofness, we derive  
    \begin{equation} \label{eq:c2}
       \mathbf{d}(c_i,\mu_i(c))\leq \mathbf{d}(c_i,\mu_i((c'_i,c_{-i})) )
    \end{equation} and
         
       \begin{equation}\label{eq:c3}
         \mathbf{d}(c'_i,\mu_i((c'_i,c_{-i})))\leq \mathbf{d}(c'_i,\mu_i(c)).  
       \end{equation} 
   \noindent  By triangle inequality, we have  
 \begin{equation} \label{eq:c4}
 \mathbf{d}(c_i,\mu_i((c'_i,c_{-i})))\leq  \mathbf{d}(c_i,c'_i) + \mathbf{d}(c'_i,\mu_i((c'_i,c_{-i}))).     
 \end{equation}
By Equations (\ref{eq:c3}) and (\ref{eq:c4}), we derive  $\mathbf{d}(c_i,\mu_i((c'_i,c_{-i})))\leq  \mathbf{d}(c_i,c'_i) + \mathbf{d}(c'_i,\mu_i(c))$, and by (\ref{eq:c1}), we have  $\mathbf{d}(c_i,\mu_i((c'_i,c_{-i})))\leq \mathbf{d}(c_i,\mu_i(c))$. Considering (\ref{eq:c2}), we conclude 
\begin{equation}\label{eq:resul}
 \mathbf{d}(c_i,\mu_i(c))=\mathbf{d}(c_i,\mu_i((c'_i,c_{-i})).   
\end{equation}
 Since $c'_i$ is between $c_i$ and $\mu_i(c)$,  by in-betweenness property,  for all $a\in ED(c)$, $ \mu_{ia}((c'_i,c_{-i}))\leq \mu_{ia}(c)$, for all $a\in ES(c)$, $ \mu_{ia}((c'_i,c_{-i}))\geq \mu_{ia}(c)$, and  for all $a\in UN(c)$, $ \mu_{ia}((c'_i,c_{-i}))= \mu_{ia}(c)$. Using equality~(\ref{eq:resul}), we derive for all $a\in A$, $\mu_{ia}(c)=\mu_{ia}((c'_i,c_{-i}))$.   
\end{proof}


\subsection{Proofs for Section~\ref{sec:Mechan}}\label{app:Mechan}

\begin{proof} Proof of Proposition~\ref{Theorem:URL-properites}.

We consider an arbitrary  sequence of agents, denoted as $\alpha$, and an arbitrary sequence of objects, denoted as $\beta$. For simplicity, instead of using the notation $URC^{\alpha,\beta}$, we will use the shorthand $URC$. We also  Let $c=(c_i)_{i\in N}$ be an arbitrary profile.

\noindent Proof of \textbf{Efficiency}. 

  It is easy to show that the outcome of the every URC mechanism  is same-sided and by Proposition~\ref{eff-same}, efficient. According to the definition of phase 1 in URC mechanisms, for every $a\in ED(c)$, for all $i\in N$, we have $URC_{ia}(c)=UR_i(c_a)$, and for all $a\not \in ED(c)$, $w_{ia}=c_{ia}$. As the uniform rule is same sided, we have for all $a\in ED(c)$, $URC_{ia}(c)\leq c_{ia}$. For all $a\not\in ED(c)$,  in phase 2 of URC, some amount is added to $w_{ia}$, and  since $w_{ia}=c_{ia}$,  we will have $c_{ia}\leq URC_{ia}(c)$. Therefore, $URC(c)$ is same-sided.
  
\vspace{0.2cm}

  \noindent Proof of \textbf{Strategy Proofness}. 
  
  First of all since URC is same-sided,   by Lemma~\ref{lem:distance-sameside}, Equation~(\ref{distance-sameside}),
  \begin{equation}\label{eq:sp1}
      \mathbf{d}(c_i,URC_i(c))=2\times \sum_{a\in ED(c)}(c_{ia}-URC_{ia}(c)). 
  \end{equation}
  
    Since URC is same-sided, we have for all $a\in ED(c)$, $URC_{ia}(c)\leq c_{ia}$. We partition $ED(c)$ into subsets $S_1$ and $S_2$ where $a\in S_1$  if and only if $c_{ia}>URC_{ia}(c)$, and $a\in S_2$  if and only if $c_{ia}=URC_{ia}(c)$.  So, (\ref{eq:sp1}) implies
    \begin{equation}\label{eq:sp2}
        \mathbf{d}(c_i,URC_i(c))=2\times \sum_{a\in S_1}(c_{ia}-URC_{ia}(c)).
    \end{equation} If $S_1=\emptyset$, then the distance is equal zero and agent $i$ cannot get better off by misreporting. So, we assume $S_1\neq \emptyset$.
    
In URC mechanisms, for every $a\in S_1$, before agent $i$ closes his tap on tank $a$, the liquid in tank $a$ is exhausted. 
    Suppose that agent $i$ misreports $c'_i$ instead of $c_i$. For every $a\in S_1$,
    \begin{itemize}
        \item if $c_{ia}\leq c'_{ia}$ then before agent $i$ closes his tap on tank $a$, the liquid in tank $a$ is exhausted. So, $URC_{ia}((c'_i,c_{-i}))=URC_{ia}(c)$. This implies $(c_{ia}-URC_{ia}(c))= (c_{ia}-URC_{ia}((c'_i,c_{-i})))$, 
        \item if $c'_{ia}< c_{ia}$ then either again before agent $i$ closes his tap on tank $a$, the liquid in tank $a$ is exhausted, and $URC_{ia}((c'_i,c_{-i}))=URC_{ia}(c)$, or  agent $i$ closes his tap exactly when he receives $c'_{ia}$, and thus $URC_{ia}((c'_i,c_{-i}))=c'_{ia}\leq URC_{ia}(c)$. This implies $(c_{ia}-URC_{ia}(c))\leq (c_{ia}-URC_{ia}((c'_i,c_{-i})))$.
    \end{itemize}
Therefore, 
\begin{equation}\label{eq:sp3}
\sum_{a\in S_1}(c_{ia}-URC_{ia}(c))\leq \sum_{a\in S_1}(c_{ia}-URC_{ia}((c'_i,c_{-i}))).    
\end{equation}

We have $\sum_{a\in A}c_{ia}-\sum_{a\in A}URC_{ia}((c'_i,c_{-i}))=1-1=0$. 

\begin{equation*}
  0=  \sum_{a\in A}c_{ia}-\sum_{a\in A}URC_{ia}((c'_i,c_{-i}))=\sum_{a\in S_1}(c_{ia}-URC_{ia}((c'_i,c_{-i})))+ \sum_{a\in A\setminus S_1}(c_{ia}-URC_{ia}((c'_i,c_{-i})))
\end{equation*}
  Because the above equality  is equal zero, there are some objects $b_1,b_2,...,b_m\in A\setminus S_1$ such that for all
  $k$, $(c_{ib_k}-URC_{ib_k}((c'_i,c_{-i})))<0$, and 
  \begin{center}
   $-\sum_{k=1...m}(c_{ib_k}-URC_{ib_k}((c'_i,c_{-i})))\geq \sum_{a\in S_1}(c_{ia}-URC_{ia}((c'_i,c_{-i})))$.   
  \end{center}
  Therefore, the total distance of $URC_{ia}((c'_i,c_{-i}))$ from $c_{ia}$ for objects $a\in S_1\cup \{b_1,b_2,...,b_m\}$ is greater than $2\times \sum_{a\in S_1}(c_{ia}-URC_{ia}((c'_i,c_{-i}))$. Using (\ref{eq:sp2}) and (\ref{eq:sp3}), we have the total distance of $URC_{ia}((c'_i,c_{-i}))$ from $c_{ia}$ for objects $a\in S_1\cup \{b_1,b_2,...,b_m\}$ is greater than $\mathbf{d}(c_i,URC_i(c))$. Therefore, agent $i$ by misreporting $c'_i$ gets further from his ideal lottery $c_i$.

 \vspace{0.2cm}
 \noindent Proof of \textbf{Replacement Monotonicity}.
 
Let  $i\in N$ be an arbitrary. Let $c'_i$ be such that  
\begin{itemize}
    \item $ED(c)=ED((c'_i,c_{-i}))$, and
    \item for all $a\in ED(c)$, $c'_{ia}\leq c_{ia}$.
\end{itemize}
Let $a\in ED(c)$ be arbitrary. According to the definition of phase 1 in URC mechanisms, for all  $j\in N$, $URC_{ja}(c)=UR_j(c_a)$, and $URC_{ja}((c'_i,c_{-i}))=UR_j(c'_a)$ where $c'_a=(c'_{ia},c_{-ia})$. It is easy to show that since $c'_{ia}\leq c_{ia}$ then for every other agents  $j\in N\setminus \{i\}$, in the uniform rule mechanism, agent $j$ can obtain more or at least an equal amount from object $a$. In other words,  for all   $j\in N\setminus \{i\}$, $UR_j(c_a)\leq UR_j(c'_a)\leq c_{ja}$.  Thus for all $a\in ED(c)$,   $(c_{ja}-URC_{ja}(c))\geq (c_{ja}-URC_{ja}((c'_i,c_{-i})))$. 
Utilizing Lemma~\ref{lem:distance-sameside}, Equation~(\ref{distance-sameside}), we have for all $j\in N\setminus\{i\}$,
\begin{center}
    $\mathbf{d}(c_j,URC_j(c))\geq \mathbf{d}(c_j,URC_j((c'_i,c_{-i}))$,
\end{center} which is equivalent to $URC_j((c'_i,c_{-i})\succeq_{c_j} URC_j(c)$.

\vspace{0.2cm}

\noindent Proof of \textbf{Non-bossiness}.

  Let $i\in N$, and $c'_i$ be such that for all $a\in ED(c)$, $URC_{ia}((c'_i,c_{-i}))=URC_{ia}(c)$.   
  Let $a\in ED(c)$ be arbitrary. There are two cases possible:
  \begin{itemize}
      \item[1.] If, when agent $i$ closes his tap on tank $a$, there exists still liquid in the tank, then closing the tap later $(c'_{ia}>c_{ia})$ or closing the tap sooner $(c'_{ia}<c_{ia})$ would result in $URC_{ia}((c'_i,c_{-i}))\neq URC_{ia}(c)$. So, we must have $c'_{ia}=c_{ia}$. Thus, for all $j\in N$, $URC_{ja}((c'_i,c_{-i}))=URC_{ja}(c)$.
      \item[2.] If, when agent $i$ closes his tap on tank $a$, already the liquid in the tank is exhausted, then 
      \begin{itemize}
          \item if he closes the tap later $(c'_{ia}>c_{ia})$, it does not effect the process, and thus for all $j\in N$, $URC_{ja}((c'_i,c_{-i}))=URC_{ja}(c)$.
          \item If he closes the tap sooner $(c'_{ia}<c_{ia})$ but not sooner than the  time when the liquid in the tank is not yet exhausted,   then again it does not effect the process, and thus for all $j\in N$, $URC_{ja}((c'_i,c_{-i}))=URC_{ja}(c)$.

          \item If he closes the tap sooner $(c'_{ia}<c_{ia})$ such that   the liquid in the tank is not yet exhausted, then the amount that agent $i$ receives in phase 1, i.e.,   $w_{ia}((c'_i,c_{-i}))$ is less than $URC_{ia}(c)$. If
          $(\sum_{j\in N\setminus\{i\}}c_{ja})\leq (1-URC_{ia}(c))$,       
          then for all $j\in N$, $URC_{ja}((c'_i,c_{-i}))=URC_{ja}(c)$, and we are done. If $(\sum_{j\in N\setminus\{i\}}c_{ja})> (1-URC_{ia}(c))$, then since agent $i$ closes sooner other agents receive more, that is $(\sum_{j\in N\setminus\{i\}}URC_{ia}((c'_i,c_{-i})))>(1-URC_{ia}(c))$. So,  the amount of object $a$ remaining for phase 2 (after phase 1) is insufficient for agent $i$ to reach  $URC_{ia}(c)$, and thus $URC_{ia}((c'_i,c_{-i}))<URC_{ia}(c)$, and the assumption $URC_{ia}((c'_i,c_{-i}))=URC_{ia}(c)$ does not hold true.
      \end{itemize}
      
  \end{itemize}

\vspace{0.2cm}

\noindent Proof of \textbf{In-Betweenness}.

\noindent Let $i\in N$ be arbitrary, and $c'_i$ be between $c_i$ and $URC_i(c)$.  
We prove
\begin{center}
for all $a\in A$, $URC_{ia}((c'_i,c_{-i}))=URC_{ia}(c)$.    
\end{center}
 
\noindent 1) For $a\in UN(c)$,  due to same-sideness of URC,  we have  $URC_{ia}(c)=c_{ia}$. Since $c'_i\in Between(c_i,URC_i(c))$, we have $c'_{ia}=c_{ia}$, and thus   $URC_{ia}((c'_i,c_{-i}))=URC_{ia}(c)$.

\vspace{0.1cm}

\noindent 2) Let $a\in ED(c)$ be arbitrary.  Since $c'_i\in Between(c_i,URC_i(c))$, we have
\begin{equation}\label{eqbet1}
   URC_{ia}(c)\leq c'_{ia}\leq c_{ia}. 
\end{equation}
 Regarding the profile $c$, two cases are possible:
 \begin{itemize}
     \item Case 1: when agent $i$ closes his tap on  tank $a$, there exists still   liquid in the tank. In this case, we have $URC_{ia}(c)= c_{ia}$.  By (\ref{eqbet1}), we have $c'_{ia}= c_{ia}$ and thus $URC_{ia}((c'_i,c_{-i}))=URC_{ia}(c)$. 
     \item Case 2: when agent $i$ closes his tap on tank $a$, the liquid is already exhausted. So, the amount of liquid $a$ that agent $i$ takes is dependent on the ideals of other agents for object $a$. In this case, due to (\ref{eqbet1}), since $ c'_{ia}\geq URC_{ia}(c)$, when agent $i$ closes his tap on the tank for  ideal $c'_{ia}$, the liquid is also exhausted already. Therefore, $URC_{ia}((c'_i,c_{-i}))=URC_{ia}(c)$.
 \end{itemize}
 Therefore, we have for all $i\in N$, for all $a\in ED(c)\cup UN(c)$
 \begin{equation}\label{beteq43}
    URC_{ia}((c'_i,c_{-i}))=URC_{ia}(c)\leq c_{ia}
\end{equation}
3) Let $a\in ES(c)$, since $c'_i$ is between $c_i$ and $URC_i(c)$, we have

\begin{equation}\label{beteq2}
    c_{ia}\leq c'_{ia}\leq URC_{ia}(c).
\end{equation}

In phase 1, for the profile $(c'_i,c_{-i})$, agent $i$ takes $w_{ia}=c'_{ia}\leq URC_{ia}(c)$ amount of object $a$, and other agents $j\in N\setminus\{i\}$ takes $w_{ia}=c_{ia}$. So, after the end of phase 1, for all $a\in ES(c)$, for every agent $i\in N$,
\begin{equation}\label{beteq44}
w_{ia}\leq  URC_{ia}(c).    
\end{equation}
Now let's analyze  phase 2 of URC, and suppose that for some $t$, $\alpha(t)=i$, and for some $s$, $\beta(s)=a$.
Let $t=1$ and $s=1$, and we are at the start of phase 2, and it is the turn of agent $\alpha(t)=i$ to take some amount of $\beta(s)=a$. Recall that in the URC mechanism $w_{jb}$ represents the current amount of object $b$ taken by agent $j$. By (\ref{beteq43}) and (\ref{beteq44}), at the current state:

\begin{itemize}
    \item[*)] Agent $i$ has taken  $(1-\sum_{b\in A\setminus\{a\}}w_{ib})\geq (1-\sum_{b\in A\setminus\{a\}}URC_{ib}(c))= URC_{ia}(c)$ of other objects except object $a$. The inequality $(1-\sum_{b\in A\setminus\{a\}}w_{ib})\geq  URC_{ia}(c)$ means that the bucket of agent $i$ has capacity to contain $ URC_{ia}(c)$ of object $a$.
    \item[**)] Other agents except $i$ has taken  $(1-\sum_{j\in N\setminus\{i\}}w_{ja}) \geq (1-\sum_{j\in N\setminus\{i\}} URC_{ja}(c)) = URC_{ia}(c)$. amount of object $a$. The inequality $(1-\sum_{j\in N\setminus\{i\}}w_{ja}) \geq  URC_{ia}(c)$ means that there exists still at least $URC_{ia}(c)$ amount of liquid $a$ totally in tank $a$ and in bucket of agent $i$.
\end{itemize} In phase 2, when agent $\alpha(t)=i$ wants to take from object $\beta(s)=a$, he takes as much as he can until either the object $a$ is exhausted (the tank $a$ is empty) or he has no capacity (his bucket is full). 

By (*) and (**), agent $i$, in his turn, can reach  $URC_{ia}(c)$ for liquid $a$, in his bucket, and thus

\begin{equation}\label{EQES1}
    URC_{ia}((c'_i,c_{-i}))\geq URC_{ia}(c).
\end{equation}
 Furthermore, we argue  $URC_{ia}((c'_i,c_{-i}))\leq URC_{ia}(c)$. We compare the bucket of agent $i$ at the beginning of phase 2, for two profiles $c$ and $(c'_i,c_{-i})$. Let $buk_{i}(c)$ refers to the bucket of agent $i$ at the beginning of phase 2 for profile $c$, and similarly $buk_{i}((c'_i,c_{-i}))$ refers to the bucket of agent $i$ at the beginning of phase 2 for profile $(c'_i,c_{-i})$. Because of (\ref{beteq43}), for every object $b\in ED(c)\cup UN(c)$, the amount of liquid $b$ in $buk_{i}(c)$ is the same as in  $buk_{i}((c'_i,c_{-i}))$, and since $c'_i$ is between $c_i$ and $URC_i(c)$, for every object $b\in ES(c)$, at the beginning of phase 2, according to the instruction of URC mechanisms, the amount of liquid $v$ in $buk_{i}(c)$ is equal to $c_{ib}$ not more than the amount of $a$ in  $buk_{i}((c'_i,c_{-i}))$ which is equal to $c'_{ib}$. So, at the beginning of phase 2, 
\begin{center}
 $buk_{i}(c)$ has more free capacity than $buk_{i}((c'_i,c_{-i}))$~ ($\#$)   
\end{center}
 Also, let $tank_a(c)$ be the status of tank $a$ running URC for profile $c$, and $tank_a((c'_i,c_{-i}))$ be the status of tank $a$ running URC for profile $(c'_i,c_{-i})$. Since $c'_i$ is between $c_i$ and $URC_i(c)$, at the beginning of phase 2,  
 \begin{quote}
   the amount of liquid in $tank_a(c)$ is not less than the the amount of liquid in $tank_a((c'_i,c_{-i}))$~ ($\#\#$)
 \end{quote}
Since  $buk_{i}(c)$ has more free capacity than $buk_{i}((c'_i,c_{-i}))$~ ($\#$) and the amount of liquid in $tank_a(c)$ is not less than the the amount of liquid in $tank_a((c'_i,c_{-i}))$~ ($\#\#$), for profile $c$, agent $i$ takes out from object $a$ not less than the amount he takes outs from object $a$ for profile $(c'_i,c_{-i})$, so we conclude  
\begin{equation}\label{EQES2}
   ~URC_{ia}((c'_i,c_{-i}))\leq URC_{ia}(c).
\end{equation}
Therefore, (\ref{EQES1}) and (\ref{EQES2}) implies:
\begin{equation}\label{eq:betq4erty} 
URC_{ia}((c'_i,c_{-i}))= URC_{ia}(c).    
\end{equation}
So, we proved for $t=1$, and $s=1$ that Equation \ref{eq:betq4erty} holds true. 
For $t=1$ and $s=2$, since we already showed Equation \ref{eq:betq4erty} for $t=1$ and $s=1$, the  conditions *) and **) yet holds true and we can repeat the same argument, and prove that For $t=1$ and $s=2$, also Equation \ref{eq:betq4erty} holds. By induction, we have for all $k=1...n$, for $t=1$ and $s=k$, Equation \ref{eq:betq4erty} holds true.
For $t>1$, since 
\begin{itemize}
    \item the ideal lottery of all agents $\alpha(1), \alpha(2),...,\alpha(t-1)$ are the same in both profiles $c$ and $(c'_i,c_{-i})$,
    \item for every excess supply object $b$, $c'_{ib}\leq \mu_{ib}(c)$,
    \item for all $j\in N$, for all $b\in ED(c)\cup UN(c)$, $URC_{ib}(c)=URC_{ib}((c'_i,c_{-i}))$, 
\end{itemize} when we run the mechanism for profile $(c'_i,c_{-i})$, agents $\alpha(1), \alpha(2),...,\alpha(t-1)$,   in their turn, take out the same amount of objects as they would for the profile $c$. Thus,  again conditions *) and **)  holds true and we can repeat the same argument for $t>1$.
Therefore, 
\begin{equation*}  for~all~a\in EC(c),~
URC_{ia}((c'_i,c_{-i}))= URC_{ia}(c).    
\end{equation*}

In this way, we proved 
\begin{equation*}  for~all~a\in A,~
URC_{ia}((c'_i,c_{-i}))= URC_{ia}(c),   
\end{equation*}
which implies that URC mechanisms satisfy in-betweenness.

 \vspace{0.2cm}

\noindent Proof of \textbf{Envy Freeness}.

The proof of envy freeness is derived from Lemma~\ref{lem:distance-sameside}, Equation~(\ref{distance-sameside}), and the fact that the uniform rule is envy free. By Equation~(\ref{distance-sameside}), we can calculate the distance based on excess demand objects, and all excess demand objects are exhausted in phase 1 of URC, and divided among the agents uniformly. Since the uniform rule is envy free, we have URC mechanisms are also envy free.

 \vspace{0.2cm}

\noindent Proof of \textbf{Anonymity}.

 The proof of anonymity is straightforward. By Lemma~\ref{lem:distance-sameside}, Equation~(\ref{distance-sameside}),  we can calculate the distance based on excess demand objects, and all excess demand objects are exhausted in phase 1 of URC. So, the distances are independent on phase 2 of URC mechanisms, and permutations does not effect phase 1 of URC mechanisms. 
\end{proof}

\begin{proof} Proof of Proposition~\ref{prop:ESimpossible}.
 
\noindent  Let $|A|=|N|=n\geq 3$, and $a,b\in A$. For every agent $j\in N$, let $z^j$ be a profile where
    \begin{itemize}
        \item $ES(z^j)=\{a\}$, 
        \item for all $i\in N$, $z^j_{ia}=0$, 
        \item $z^j_{jb}=1$ and $b\in UN(z^j)$.
    \end{itemize}
   Let's assume, for the sake of contradiction, that there exists an efficient mechanism $\mu$ satisfying the following assumption:
              \begin{quote}
           for every two profiles $c$ and $c'$ with $ES(c)=ES(c')$, that coincide on $ES(c)$, we have for every $i\in N$, $\mathbf{d}(\mu_i(c),c_i)=\mathbf{d} (\mu_i(c'),c'_i)$.
       \end{quote}
        Because $\mu$ is efficient (same-sided, as per Proposition~\ref{eff-same}),  and $b\in UN(z^j)$, we have $\mu_{jb}(z^j)=z^j_{jb}=1$ implying
    $\mathbf{d}(\mu_j(z^j), z^j_j) = 0$, implying $\mu_j(z^j)= z^j_j$, specifically $\mu_{ja}(z^j)= z^j_{ja} = 0$.
    
       Consider an arbitrary profile $x$ with $ES(x)=ES(z^j)=\{a\}$, where $x$ coincides with $z^j$ on the set~$\{a\}$.     
       As $x$ and $z^j$ coincide on $ES(z^j)$, according to the  assumption for the mechanism $\mu$, we conclude  $\mathbf{d} (\mu_j(x),x_j)=\mathbf{d}(\mu_j(z^j),z^j_j)=0$. leading to $\mu_j(x) = x_j$, especially $\mu_{ja}(x) = x_{ja} = 0$. So,
     \begin{quote}
        for every  profile $x$ with $ES(x)=ES(z^j)=\{a\}$, where $x$ coincides with $z^j$ on the set $\{a\}$ we have $\mu_{ja}(x) = x_{ja} = 0$.  
     \end{quote}
       Let $x$ be a profile such that $ES(x)=\{a\}$, and for all agent $i\in N$, $x_{ia}=0$. For every agent $j\in N$, the profile $x$ coincides with $z^j$ on the set $\{a\}$, and thus we have:  for every $j\in N$, $\mu_{ja}(x) = x_{ja} = 0$ which implies $\sum_{i\in N}\mu_{ia}=0<1$. It contradicts with object feasibility of the mechanism $\mu$. 
 \end{proof}

\subsection{Proofs for Section~\ref{sec:charact}}\label{app:charact}

\begin{proof} Proof of Lemma~\ref{prop-CSP}:

      \noindent 
     Let $c$ and $c'$ be two profiles   such that
     \begin{itemize}
         \item[]$ED(c)=ED(c')=\{a\}$, 
         \item[] for an agent $j\in N$, for all $b\in ED(c)$, $c_{jb}=c'_{jb}$, and
         \item[]  for all agents $i\in N\setminus\{ j\}$, $c'_i=c_i$, i.e. $c'=(c'_j,c_{-j})$.
     \end{itemize}  
           Since $\mu$ is strategy proof, we have
      \begin{equation}\label{eq:CSP1}
        \mathbf{d}(c_j,\mu_j(c))\leq \mathbf{d}(c_j,\mu_j(c')).   
      \end{equation}
        
   Furthermore, due to (\ref{eq:CSP1}) and having that  for all objects  $b\in ED(c)$,  $c'_{jb}\leq c_{jb}$,  we can apply Definition~\ref{DEF:RM} of replacement monotonicity, and conclude as follows:
        \begin{equation}\label{eq:CSPQ1}
           for~ all ~i\in N\setminus\{j\}, ~\mathbf{d}(c_i,\mu_i(c))\geq \mathbf{d}(c_i,\mu_i(c')).   
        \end{equation}
              Similarly, because $\mu$ is strategy proof, we have 
      \begin{equation}\label{eq:CSP2}
        \mathbf{d}(c'_j,\mu_j(c'))\leq \mathbf{d}(c'_j,\mu_j(c)).  
      \end{equation}
        Since for all objects  $b\in ED(c')$,  we have $c_{jb}\leq c'_{jb}$, and due to (\ref{eq:CSP2}), we can apply the replacement monotonicity property, leading to the following conclusion:
    
    \begin{equation}\label{eq:CSPQ2}
         for ~all ~i\in N\setminus\{j\},~\mathbf{d}(c_i,\mu_i(c'))\geq \mathbf{d}(c_i,\mu_i(c)).
    \end{equation}
             Due to (\ref{eq:CSPQ1}) and (\ref{eq:CSPQ2}), we have:    
  \begin{equation}\label{eq:CSP3}
       for~ all~ i\in N\setminus\{j\},  ~\mathbf{d}(c_i,\mu_i(c'))=\mathbf{d}(c_i,\mu_i(c)). 
  \end{equation} 
    
\noindent By efficiency and using Lemma~\ref{lem:distance-sameside}, Equation~(\ref{distance-sameside}), we have for all $i\in N\setminus\{j\}$, $\mathbf{d}(c_i,\mu_i(c))=2\times \sum_{b\in ED(c)} (c_{ib}-\mu_{ib}(c))$ and $\mathbf{d}(c_i,\mu_i(c'))=2\times \sum_{b\in ED(c')}  (c'_{ib}-\mu_{ib}(c'))$. Since for all $i\neq j$, $c'_i=c_i$, and $ED(c)=ED(c')$,  using (\ref{eq:CSP3}), we derive 

\begin{equation}\label{eq:CSP4}
    for~ all ~i\in N\setminus\{j\}, ~ \sum_{b\in ED(c)}\mu_{ib}(c)=\sum_{b\in ED(c)}\mu_{ib}(c').
\end{equation}
By object feasibility and (\ref{eq:CSP4}),   we have 
for all $i\in N$, $\sum_{b\in ED(c)}\mu_{ib}(c)=\sum_{b\in ED(c)}\mu_{ib}(c')$.
Since $ED(c)=ED(c')=\{a\}$, we have   
for all $i\in N$, $\mu_{ia}(c)=\mu_{ia}(c')$. 
\end{proof}

   \begin{proof} Proof of Lemma~\ref{lem:Twoprofiles}. 
 
\noindent  Let $c$ and $c'$ be two profile with $ED(c)=ED(c')=\{a\}$, and   $c_a=c'_a$. 
We start from the profile $c^0=c$ and construct a sequence of profiles  \begin{center}
         $c=c^0,c^1,c^2,c^3,\ldots,c^m=c'$ 
     \end{center} such that for all $t<m$, $c^{t}$ and $c^{t+1}$ satisfies conditions of Lemma~\ref{prop-CSP} for some agent $j_t\in N$, and thus we have  for all $i\in N$, $\mu_{ia}(c^t)=\mu_{ia}(c^{t+1})$.

Visualize each  non-excess demand object (objects in the set $A\setminus\{a\}$) as a warehouse with a capacity 1. In this setup, we have a collection of warehouses: $A\setminus\{a\}=\{w_1,w_2,\ldots,w_{n-1}\}$. We imagine that each agent has occupied warehouses with their items, and each profile corresponds to an occupation of warehouses with the agent's items.

We start at the initial profile $c^0=c$, where for each $k\in \{1,2,\ldots,n-1\}$, 
each agent $i\in N$  has occupied a quantity $c_{iw_k}$ of the warehouse $w_k$.
 We denote the available capacity of each warehouse $w_k$, by $Q(w_k)$. We say a warehouse does not have available capacity whenever $Q(w_k)= 0$. Initially, for every $k\in \{1,2,\ldots,n-1\}$, $Q(w_k)=1-\sum_{i\in N}c_{iw_k}$. Note that since the only excess demand object of $c^0=c$ is object $a$, we have for all $k$, $Q(w_k)\geq 0$. The following algorithm explains   how we convert profile $c$ to $c'$ by asking agents to move their items between warehouses.  
 \begin{itemize}
     \item  Initially, for each $k\in \{1,2,\ldots,n-1\}$, each agent $i\in N$  has occupied a quantity $c_{iw_k}$ of the warehouse $w_k$.
     \item  At the end of the algorithm, for each $k\in \{1,2,\ldots,n-1\}$, each agent $i\in N$  has occupied a quantity $c'_{iw_k}$ of the warehouse $w_k$.
 \end{itemize}

\begin{itemize}
\item[0.] Set $t=1$ and $s=1$.

    \item[1.] For agent $t$, we  want to change $c_{tw_s}$  to $c'_{tw_s}$.
    \begin{itemize}
        \item[1.1.] If $c_{tw_s}>c'_{tw_s}$ then we ask agent $t$ to  move the excess, denoted as $\epsilon=(c_{tw_s}-c'_{tw_s})$, from the warehouse $w_s$ to warehouses  $w_{s+1},w_{s+2},\ldots,w_{n-1}$. To do this, we ask him first transfer an amount $v=\min(\epsilon,\sum_{k=s+1}^{n-1}Q(w_k))$ from $w_s$ to warehouses  $w_{s+1},w_{s+2},\ldots,w_{n-1}$. If $v<\epsilon$,   agents $t+1,t+2,\ldots, n$ are sequentially requested to move their items as much as possible from warehouses 
 $w_{s+1},w_{s+2},...,w_{n-1}$ to those warehouses $w_1,w_2,...,w_s$ that have available capacity, effectively freeing up space in warehouses $w_{s+1},w_{s+2},...,w_{n-1}$. Then agent $t$ proceeds to move as much as possible from  $w_s$ to warehouses  $w_{s+1},w_{s+2},...,w_{n-1}$. This process continues until agent $i$ has the quantity $c'_{tw_s}$ in the warehouse $w_s$.
 \item[1.2.] If $c_{tw_s}<c'_{tw_s}$ then we ask agent $t$ to make up for the shortfall, represented as $\epsilon=(c'_{tw_s}-c_{tw_s})$,  by moving items from warehouses $w_{s+1},w_{s+2},\ldots,w_{n-1}$  to warehouse $w_s$. He moves items from $w_{s+1},w_{s+2},\ldots,w_{n-1}$ to $w_s$ until one of the following occurs: either he has moved $\epsilon$ items, or, warehouses $w_{s+1},w_{s+2},\ldots,w_{n-1}$ are empty, or $w_s$ has no available capacity. In the first case, the goal is achieved.  The second case does not happen, 
 because $\sum_{k\geq s}c_{tw_k}=\sum_{k\geq s}c'_{tw_k}$, and since $c_{tw_s}<c'_{tw_s}$, $w_{s+1},w_{s+2},\ldots,w_{n-1}$ cannot be empty.
 If the third case occurs,   agents $t+1,t+2,\ldots, n$ are sequentially requested to move their items as much as possible from warehouses 
 $w_{s}$ to warehouses $w_{s+1},w_{s+2},...,w_{n-1}$, effectively freeing up space in warehouses $w_{s}$. Then agent $t$ proceeds to move his items as much as possible from  $w_{s+1},w_{s+2},...,w_{n-1}$  to the warehouse  $w_s$. This process continues until agent $i$ has the quantity $c'_{tw_s}$ in the warehouse $w_s$.

 \item[1.3.] If $s<n-1$, update $s$ to $s+1$, go to Step 1.1. 

 \end{itemize}

\item[2.] If $t<n$, update $t$ to $t+1$, go to Step 1.
 
    \end{itemize}
  For $s=n-1$, since agent $t$ has already changed $c_{tw_s}$ to $c'_{tw_s}$ for all $s<n-1$, according to agent feasibility, the amount of his items that finally ends up in warehouse $w_{n-1}$ is equal to $c'_{tw_{n-1}}$.

For $t=n$, since all agents $1, 2, \ldots, n-1$ have already changed for all $s$, $c_{tw_s}$ to $c'_{tw_s}$, when agent $n$ wants to move his items according to Step 1.1 and Step 1.2, he does not face a situation where the warehouses do not have available capacity.
    
During execution of the algorithm, each  $(c_{ia},(w_{ik})_{1\leq k\leq n-1})_{i\in N}$, which corresponds to
an occupation of warehouses with the agent’s items, forms a profile. 
    These profiles are updated during the execution of the algorithm, gradually converging to the final profile $c'$. Note that in the above algorithm, always one agent move his items between warehouses which are representing non-excess demand objects. So, in above algorithm, every time that an agent moves his items between warehouses conditions of Lemma~\ref{prop-CSP} are fulfilled. Therefore, thanks to Lemma~\ref{prop-CSP}, for all agents $i \in N$, the value of $\mu_{ia}$ remains unaltered for all  profiles $(c_{ia},(w_{ik})_{1\leq k\leq n-1})_{i\in N}$ that occur during the execution of the above algorithm. 
    Therefore, we have for all $i\in N$, $\mu_{ia}(c)=\mu_{ia}(c')$.
 \end{proof}

 \begin{remark}
     One may note that in the proofs of Lemma~\ref{prop-CSP} and Lemma~\ref{lem:Twoprofiles}, we did not make use of the assumption that $ED(c)$ has only one element. Therefore, if $\mu$ is a strategy proof, efficient, and replacement monotonic mechanism, then for any two profiles $c$ and $c'$ with $ED(c)=ED(c')$ and for all $a\in ED(c)$, $c_a=c'_a$, we have the following result: for all $i\in N$, 
     
$\sum_{a\in ED(c)}\mu_{ia}(c)=\sum_{a\in ED(c)}\mu_{ia}(c')$, and  
$\mathbf{d}(c_i,\mu_i(c))=\mathbf{d}(c_i,\mu_i(c'))$.

 \end{remark}

\begin{proof} Proof of Lemma~\ref{lemEDa}.

Let $\mu$ be a strategy proof, efficient, replacement monotonic and anonymous mechanism.
Let $c$ be a profile with $ED(c)=\{a\}$.
By utilizing Lemma~\ref{lem:Twoprofiles}, it becomes evident that for any other profile $c'$ with   $c'_a=c_a$ and $ED(c')=ED(c)={a}$, for every agent $i\in N$, $\mu_{ia}(c')=\mu_{ia}(c)$. So, for every agent $i$, his allocation for object $a$ is independent on the ideals for non-excess demand objects and only is dependent on $c_a$. This means that the function $\pi_a\circ\mu~(-)= (\mu_{ia}(-))_{i\in N}$ is a well-defined function of $c_a$, where  $\pi_a$ is a projection mapping on the allocation for object $a$.
 Since $\mu$ is same-sided, strategy proof   and anonymous, these properties extend to $\pi_a\circ\mu$.  According to the characterization theorem proved by Y. SPrumont, (1999) in page 511 of~\cite{SPRUMONT1991Econometrica}, $\pi_a \circ \mu$ is the uniform rule, and thus for all $i\in N$, $\mu_{ia}(c)=UR_i(c_a)$. Hence, by Lemma~\ref{lem:distance-sameside}, Equation~(\ref{distance-sameside}), for every profile $c$ with $ED(c)=\{a\}$, we have for all $i\in N$, $\mathbf{d}(c_i,\mu_i(c))=2\times (c_{ia}- UR_i(c_a))$.

On the other hand, because URC mechanisms are efficient, utilizing Lemma~\ref{lem:distance-sameside}, Equation~(\ref{distance-sameside}), and  definition of phase 1 of URC mechanisms,   for all $i\in N$,
$\mathbf{d}(c_i,URC_i(c))=2\times (c_{ia}- UR_i(c_a))$.  Therefore, $\mathbf{d}(c_i,\mu_i(c))=\mathbf{d}(c_i,URC_i(c))$,  and we are done. \end{proof}

 
\begin{proof} Proof of  Theorem~\ref{charac}.

Let $\mu$ be a strategy proof, efficient, replacement monotonic, non-bossy, in-between and anonymous mechanism. Let $c=(c_i)_{i\in N}$ be arbitrary profile. Utilizing Corollary~\ref{cor:between},   and Lemma~\ref{lemEDa}, we prove for all $b\in ED(c)$, for all $i\in N$, $\mu_{ib}(c)=UR_i(c_b)$ where $c_b=(c_{ib})_{i\in N}$. Thus for all $i\in N$,  $\mu_i(c)\equiv_{c_i} URC_i(c)$.

\noindent Let $a\in ED(c)$ be arbitrary.   we transform  the profile $c$ into another profile called $z^n$ featuring only a single excess demand object, $ED(z^n)=\{a\}$, and  for every $i\in N$, $\mu_{ia}(c)=\mu_{ia}(z^n)$ .

 For agent $1\in N$, we let
    \begin{center}
        $c'_1=(c_{1a},(\mu_{1b}(c))_{b\in ED(c)\setminus\{a\}},(c_{1b})_{b\in UN(c)},(l_{1b})_{b\in ES(c)}))$
    \end{center} where for each $b\in ES(c)$,   $c_{1b} \leq l_{1b}\leq \mu_{1b}(c)$. The lottery   $c'_1$ is between $c_1$ and $\mu_1(c)$.
    
    We replace $c_1$ with $c'_1$ in profile $c$ and construct the profile $z^1=(c'_1,c_2,c_{3},\ldots c_{n})$. 
By Proposition~\ref{prop:Between}, 
    
\begin{equation}\label{eq:mt1}
  for ~all ~b\in A~ \mu_{1b}(c)=\mu_{1b}(z^1).
\end{equation}
    Since $\mu$ is same-sided, we have
 \begin{equation}\label{eq:mtt}
     ED(z^1)\subseteq ED(c),~ ES(z^1)\subseteq ES(c),~
          a\in ED(z^1), ~z^1_a=c_a 
 \end{equation}
  
   By  non-bossiness (Section~\ref{nonbos}), and (\ref{eq:mt1})  we have 
 \begin{equation}\label{eq:mt2}
  for~ all~j\in N,~   for~ all~b\in ED(c),~\mu_{jb}(z^1)=\mu_{jb}(c).   
 \end{equation}
  
 We repeat the above process on the profile $z^1$ for agent $2\in N$, and construct a lottery $c'_2$ (similar to the way that we construct $c'_1$)
  \begin{center}
        $c'_2=(z^1_{2a},(\mu_{2b}(z^1))_{b\in ED(z^1)\setminus\{a\}},(z^1_{2b})_{b\in UN(z^1)},(l_{2b})_{b\in ES(z^1)}))$
    \end{center} where for each $b\in ES(z^1)$,   $c_{2b} \leq l_{2b}\leq \mu_{2b}(z^1)$. The lottery $c'_2$ is between $c_2$, and $\mu_2(z^1)$. Note that because of (\ref{eq:mt2}) and (\ref{eq:mtt}), we have 
 \begin{center}
        $c'_2=(c_{2a},(\mu_{2b}(c))_{b\in ED(z^1)\setminus\{a\}},(z^1_{2b})_{b\in UN(z^1)},(l_{2b})_{b\in ES(z^1)}))$.
    \end{center}
    
We replace $z^1_2$ with $c'_2$ in profile $z^1$ and construct the profile $z^2=(c'_1,c'_2,c_{3},\ldots c_{n})$.  since $\mu$ is same-sided $ED(z^2)\subseteq ED(z^1)$.   According to the definition of the lottery $c'_2$, and (\ref{eq:mt2}), we have
      \begin{center}
          for agents $i=1,2$, for all objects
    $b\in ED(z^2)\setminus\{a\}$,  
    $z^2_{ib}=\mu_{ib}(c)$.
      \end{center}
Also it is easy to check that $z^2_a=z^1_a=c_a$. 

Since $c'_2$ is between $c_2$, and $\mu_2(z^1)$, by Proposition~\ref{prop:Between}, we derive for all $b\in A$, $\mu_{2b}(z^1)=\mu_{2b}(z^2)$.
With a similar line of reasoning, and by leveraging  Corollary~\ref{cor:between},   we can demonstrate  
\begin{equation}\label{eq:mt4}
 for~ all~i\in N,~ for~ all~b\in ED(z^1),~\mu_{ib}(z^2)=\mu_{ib}(z^1).
 \end{equation}
   If we repeat the above process for all agents $k\in N=\{1,2,...,n\}$,
    we construct profiles
    \begin{center}
      $z^0=c, z^1,z^2,...,z^n$   
    \end{center}
       where for $k<n$:     $z^{k+1}_a=z^k_a=c_a$, $a\in  ED(z^{k+1})\subseteq ED(z^{k})$,  and
 
\begin{equation}\label{eq:mt10}
 for~ all~i\in N,~ for~ all~b\in ED(z^k),~\mu_{ib}(z^{k+1})=\mu_{ib}(z^{k}).
 \end{equation}
 By (\ref{eq:mt10}) and the definitions of profiles $z^1,z^2,...,z^n$, we have for all $k<n$,
\begin{equation}\label{eq:mt5}
for~ agents~ i=1,2,...,k+1,~ for~ all ~
   b\in ED(z^{k+1})\setminus\{a\},~    z^{k+1}_{ib}=\mu_{ib}(c).  
\end{equation}
          For the profile $z^n$ we prove  $ED(z^n)=\{a\}$. Suppose that $d\in ED(z^n)\setminus \{a\}$. By (\ref{eq:mt5}), we have for all $i\in N$, $z^n_{id}=\mu_{id}(c)$. Thus $\sum_{i\in N}z^n_{id}=\sum_{i\in N}\mu_{id}(c)=1$, hence $d\in UN(z^n)$, contradiction. Also, by (\ref{eq:mt10}), we have for every $i\in N$, $\mu_{ia}(c)=\mu_{ia}(z^n)$ .
    

Since  $ED(z^n)=\{a\}$,  by Lemma~\ref{lemEDa}, we have for every $i\in N$,
    $\mu_{ia}(z^n)=UR_i(z^n_a)$. By (\ref{eq:mt10}), for all $i\in N$,
    $\mu_{ia}(c)=\mu_{ia}(z^n)$, thus we conclude $\mu_{ia}(c)=UR_i(c_a)$ for all $i\in N$.

\noindent    As $a\in ED(c)$ is assumed arbitrary, the same   argument holds for all  objects in $ED(c)$ and thus
    \begin{center}
        for all $b\in ED(c)$, for all $i\in N$, $\mu_{ib}(c)=UR_i(c_b)$.
    \end{center} Since   $\mathbf{d}(c_i, \mu_i(c))=2\times\sum_{o\in ED(c)}(c_{io}-\mu_{io}(c))$, we have $\mu$ is welfare equivalent to URC mechanisms, that is for every profile $c$, for every agent $i\in N$, $\mathbf{d}(c_i,\mu_i(c))=\mathbf{d}(c_i,URC_i(c))$.
\end{proof}

 \subsection{Proofs for Section~\ref{sec:indaxiom}}\label{app:indaxion}

\begin{proof} Proof of Proposition~\ref{prop:PDSD}.

\noindent \textbf{Proof of Properties for SDC Mechanisms:} 

Suppose a profile $c=(c_i)_{i\in N}$ is given. Let $\alpha$ be a sequence of agetns and $\beta$ a sequence of objects. In phase 1 of SDC mechanisms, when $w_{\alpha(t)a}=\min(r_a,c_{\alpha(t)a})$, if $a\in ED(c)$, then the object is exhausted in phase 1, and thus nothing remains for phase 2, and for every $i\in N$, $SDC^{\alpha, \beta}_{\alpha(t)a}(c)= w_{\alpha(t)a}\leq c_{\alpha(t)a}$. If $a\in ES(c)\cup UN(c)$, then in phase 1, $w_{\alpha(t)a}=\min(r_a,c_{\alpha(t)a})=c_{\alpha(t)a}$, and thus $SDC^{\alpha, \beta}_{\alpha(t)a}(c)\geq c_{\alpha(t)a}$.
This is an evident that the outcome is same-sided and, therefore, efficient by Proposition~\ref{eff-same}.

The proof that SDC mechanisms are strategy-proof is straightforward. In phase 1, each agent, during their turn, tries to get as close as possible to their peaks.  Note that for all $a\in ED(c)$,  $SDC^{\alpha, \beta}(c)$ is determined in phase 1.
For $a\in ED(c)$, if $c_{\alpha(t)a}\leq r_a$, then $w_{\alpha(t)a}=c_{\alpha(t)a}$, and if $c_{\alpha(t)a}> r_a$ then $w_{\alpha(t)a}=r_a$. Suppose   agent $\alpha(t)$ misreport $c'_i$. 
Then it is easy to see that for all $a\in ED(c)$, $(c_{\alpha(t)a}-SDC^{\alpha, \beta}_{\alpha(t)a}(c))\leq (c_{\alpha(t)a}-SDC^{\alpha, \beta}_{\alpha(t)a}((c'_i, c_{-i}))$. Thus, by Lemma~\ref{lem:distance-sameside}, Equation~(\ref{distance-sameside}), misreporting does not result in a decrease in distance.

The SDC mechanisms are replacement monotonic. Let $c'_i$ be such that  
\begin{itemize}
    \item $ED(c)=ED((c'_i,c_{-i}))$, and
    \item for all $a\in ED(c)$, $c'_{ia}\leq c_{ia}$.
\end{itemize}
Let $s$ be the index such that $\alpha(s)=i$.  For all $t<s$,  for all $a\in A$, $SDC^{\alpha, \beta}_{\alpha(t)a}(c)=SDC^{\alpha, \beta}_{\alpha(t)a}((c'_i, c_{-i}))$. This is because the change made by agent $i$ from $c_i$ to $c'_i$ does not affect the agents who precede him in the sequence.
In the case where $t > s$, consider the remainder of object $a$ at agent $\alpha(t)$'s turn in phase 1, denoted as $r_a(t)$. Also
  let $r'_a(t)$ be the remainder of object $a$ at agent $\alpha(t)$'s turn in phase 1, when agent $\alpha(s)=i$ changed  from $c_i$ to $c'_i$. 
For all $a \in ED(c)$, since $c'_{ia} \leq c_{ia}$, we have $r_a(t) \leq r'_a(t)$. This enables agent $\alpha(t)$ to still have the opportunity to get as close to his peak as he did before changing from $c_i$ to $c'_i$. Consequently, for all $j\in N\setminus\{i\}$, $\mathbf{d}(c_j,\mu_j(c))\geq \mathbf{d}(c_j,\mu_j((c'_i,c_{-i})))$.

To prove non-bossiness, let's consider an agent $\alpha(t)$ who intends to behave in a bossy manner. This agent cannot affect those agents who precede him in the sequence $\alpha$. Additionally, if agent $\alpha(t)$ does not take a different amount of excess demand objects in his turn, then he does not affect the remainder of excess demand objects for agents who succeed him in the sequence $\alpha$. Consequently, all agents $\alpha(s)$, where $s>t$, receive the same amount for excess demand objects as they would have received without the interference of agent $\alpha(t)$. 

To prove in-betweenness, suppose for an arbitrary agent $i=\alpha(t)$, $c'_{\alpha(t)}$ is a lottery between $c_{\alpha(t)}$ and $SDC^{\alpha,\beta}_{\alpha(t)}(c)$.  Then as SDC mechanisms are same-sided, we have for all $a\in A$,
\begin{itemize}
    \item if $a\in UN(c)$, then $c'_{\alpha(t)a}=c_{\alpha(t)a}=SDC^{\alpha,\beta}_{\alpha(t)a}(c)$,
    \item if $a\in ED(c)$, then $SDC^{\alpha,\beta}_{\alpha(t)a}(c)\leq c'_{\alpha(t)a}\leq c_{\alpha(t)a}$, and
    \item if $a\in ES(c)$, then $SDC^{\alpha,\beta}_{\alpha(t)a}(c)\geq c'_{\alpha(t)a}\geq c_{\alpha(t)a}$.
\end{itemize}
   In phase 1 of SDC mechanisms, after changing from  $c_{\alpha(t)}$ to $c'_{\alpha(t)}$, agent $\alpha(t)$ takes

   \begin{center}
  $w'_{\alpha(t)a}=\min(r_a,c'_{\alpha(t)a})$.     
   \end{center}
If $a\in ED(c)\cup UN(c)$, then since $SDC^{\alpha,\beta}_{\alpha(t)a}(c)=\min(r_a,c_{\alpha(t)a})$, and  $SDC^{\alpha,\beta}_{\alpha(t)a}(c)\leq c'_{\alpha(t)a}\leq c_{\alpha(t)a}$, we have

\begin{equation*}\label{eqSDC1}
SDC^{\alpha,\beta}_{\alpha(t)a}((c'_{\alpha(t)},c_{{-\alpha(t)}}))=w'_{\alpha(t)a}=\min(r_a,c'_{\alpha(t)a})=\min(r_a,c_{\alpha(t)a})=SDC^{\alpha,\beta}_{\alpha(t)a}(c).    
\end{equation*}
Thus,
\begin{equation}\label{betSDC}
 for~ all~a\in ED(c)\cup UN(c),~ SDC^{\alpha,\beta}_{\alpha(t)a}((c'_{\alpha(t)},c_{{-\alpha(t)}}))=SDC^{\alpha,\beta}_{\alpha(t)a}(c).  
\end{equation}
By   non-bossiness, the allocation for excess demand objects does not change for all agents.

As the phase 2 of SDC mechanisms resembles the phase 2 of URC mechanisms, and also similar to URC mechanisms, the outcome of SDC mechanisms remains unchanged for excess demand and unanimous objects  (see \ref{betSDC}), therefore, the same proof  demonstrated for URC mechanisms, which establishes that for all $a\in ES(c)$, $
URC_{ia}((c'_i,c_{-i}))= URC_{ia}(c)$, (see \ref{eq:betq4erty}), works for proving ``for all $a\in ES(c)$, $SDC^{\alpha,\beta}_{\alpha(t)a}((c'_{\alpha(t)},c_{{-\alpha(t)}}))=SDC^{\alpha,\beta}_{\alpha(t)a}(c)$". Hence, for all $a\in A$, $
SDC_{ia}((c'_i,c_{-i}))= SDC_{ia}(c)$, and we are done.

\vspace{0.2cm}

\noindent The proof that SDC mechanisms are neither anonymous nor envy free is straightforward.

\noindent \textbf{Proof of Properties for PDC Mechanisms:}

It is easy to verify that the outcome of PDC mechanisms is same-sided and thus efficient by  Proposition~\ref{eff-same}. The proof of anonymity for PDC mechanisms is also straightforward. 

The proof of non-bossiness for PDC mechanisms is straightforward. Since excess demand objects are allocated proportionally, if an agent wants to not change his own allocation for excess demand objects, then he should not change his ideals for excess demand objects, and thus the allocations of excess demand objects for other agents remain unchanged.

  The PDC mechanisms are replacement monotonic. Let $c'_i$ be such that  \begin{itemize}
    \item $ED(c)=ED((c'_i,c_{-i}))$, and
    \item for all $a\in ED(c)$, $c'_{ia}\leq c_{ia}$.
\end{itemize}
 Since, for all $a\in ED(c)$, $c'_{ia}\leq c_{ia}$, and the allocation of excess demand objects is proportional, it follows that other agents  $j\in N\setminus\{i\}$, obtain an equal or larger proportion. Consequently, applying Lemma~\ref{lem:distance-sameside}, Equation~(\ref{distance-sameside}), we have  for all $j\in N\setminus\{i\}$, $\mathbf{d}(c_j,\mu_j(c))\geq \mathbf{d}(c_j,\mu_j((c'_i,c_{-i})))$.

\vspace{0.1cm}

PDC mechanisms satisfy in betweenness.   Let $\alpha$ be a sequence of agents and $\beta$ a sequence of objects, both arbitrary. Suppose that for some agent $i$,  a lottery $c'_i$ is between $c_i$ and $PDC^{\alpha,\beta}_i(c)$. As in phase 1 of PDC mechanisms, excess demand objects and unanimous objects are allocated proportionally, and for  every $a\in ED(c)\cup UN(c)$, $PDC^{\alpha,\beta}_{ia}(c)\leq c'_{ia}\leq c_{ia}$. Therefore, after accomplishing phase 1, we have
\begin{itemize}
    \item for all $a\in UN(c)$, $PDC^{\alpha,\beta}_{ia}((c'_i,c_{-i}))= PDC^{\alpha,\beta}_{ia}(c)$, 
    \item for all $a\in ED(c)$, $PDC^{\alpha,\beta}_{ia}((c'_i,c_{-i}))\leq PDC^{\alpha,\beta}_{ia}(c)$,  and 
    \item   since for all $a\in ES(c)$, $c_{ia}\leq c'_{ia}\leq \mu_{ia}(c)$, $w_{ia}=c'_{ia}\leq \mu_{ia}(c)$.
\end{itemize}
 Additionally, because for  every $a\in ED(c)$, $c'_{ia}\leq c_{ia}$ and $PDC^{\alpha,\beta}_{ia}(c)\leq c'_{ia}$, we have  for all other agents $j\in N\setminus \{i\}$, for  every $a\in ED(c)$,  $PDC^{\alpha,\beta}_{ja}((c'_i,c_{-i}))\geq PDC^{\alpha,\beta}_{ja}(c)$. 

 Let $v= (\sum_{a\in ED(c)}PDC_{ia}(c)-\sum_{a\in ED(c)}PDC_{ia}((c'_i,c_{-i})))$ be the amount that agent $i$ takes less from excess demand objects, when changing from $c_i$ to $c'_i$. Since PDC is proportional for excess demand objects and  for  every $a\in ED(c)$, $c'_{ia}\leq c_{ia}$ and $PDC^{\alpha,\beta}_{ia}(c)\leq c'_{ia}$, we have $v$ is also equal to the amount that other agents takes out more from excess demand objects,  when changing from $c_i$ to $c'_i$, that is, $v= \sum_{j\in N\setminus\{i\}}(\sum_{a\in ED(c)}PDC_{ja}((c'_i,c_{-i}))-\sum_{a\in ED(c)}PDC_{ja}(c))$.

We prove by contradiction that for every $a\in ES(c)$, $PDC_{ia}((c'_i,c_{-i}))\geq PDC_{ia}(c)$. Suppose not and there exists an object $b\in ES(c)$ such that $PDC_{ib}((c'_i,c_{-i}))< PDC_{ib}(c)$. 

Suppose that for the sequence of objects $\beta$, for some number $k$, $\beta(k)=b$. Since $b\in ES(c)$, already in phase 1, agent $i$ has taken $c'_{ib}\leq PDC_{ib}(c)$ amount of object $b$. In phase 2 of running $PDC$ for the profile $(c'_i,c_{-i})$, when it is the turn of agent $i$ to take some amount of object $b$, the inequality $PDC_{ib}((c'_i,c_{-i}))< PDC_{ib}(c)$ implies that agent $i$ could not take as much  of object $b$ as he would in profile $c$, so 
\begin{itemize}
    \item[case1:] We are either facing a shortage of object $b$, and the tank containing object $b$ has less liquid compared to the time of profile $c$. This scarcity results in agent $i$ being unable to take as much of object $b$ to reach $PDC_{ib}(c)$.
    \item[case2:] or the bucket of agent $i$ has less available capacity compared to the time of profile $c$, and it caused that agent $i$ cannot reach $PDC_{ib}(c)$.
\end{itemize}
As for the profile $(c'_i,c_{-i})$, all agents $j\in N\setminus\{i\}$ takes more from excess demand objects compared to the time of profile $c$, therefore, they have less free capacity for excess supply object $b$, and they take less from object $b$ compared to the time of profile $c$. So the case1 is impossible.

We argue that the case2 is also impossible. Suppose that the bucket of agent $i$ has less capacity compared to the time of profile $c$. 
For all objects $\beta(1),\beta(2),...\beta(k-1)$, agents in $N\setminus\{i\}$, can, at most, collectively take $v$ amount less than the objects $\beta(1),\beta(2),...\beta(k-1)$ compared to the time of profile $c$. So agent $i$ can take at most $v$ amount more than than the objects $\beta(1),\beta(2),...\beta(k-1)$ compared to the time of profile $c$, and since he has already takes $v$ amount less than excess demand objects, his bucket has enough capacity  to reach $PDC_{ib}(c)$ for object $b$.

Thus both case1 and case2 are impossible, contradiction. Therefore, for every $a\in ES(c)$, $PDC_{ia}((c'_i,c_{-i}))\geq PDC_{ia}(c)$. 
\vspace{0.2cm}
 
 The proof that PDC mechanisms is not strategy proof is straightforward. For envy freeness, consider $N=\{1,2,3,4\}$ and $A=\{a,b,d,e\}$. Let $c$ be preference profile where $ED(c)=\{a\}$, and $c_{1a}=c_{2a}=c_{3a}=1$, and $c_{4a}=1/3$. In phase 1 of PDC mechanisms, object $a$ is divided proportionally as follows:
$p_{1a}=p_{2a}=p_{3a}=0.3$ and $p_{4a}=0.1$. Since PDC mechanisms are same-sided, because of Lemma~\ref{distance-sameside}, we can compute the welfare based on excess demand objects.  Agent 4 envy agent 1, since $|c_{4a}-p_{1a}|<|c_{4a}-p_{4a}|$. 
\end{proof}
 
\end{document}